\documentclass[dvipsnames,moor,sglanonrev,thm-restate]{article}
\usepackage{booktabs} 
\usepackage[ruled]{algorithm2e} 

\usepackage[font=footnotesize]{subcaption}
\usepackage{amsmath,amsfonts,amssymb,amsthm}
\usepackage{tikz}
\usetikzlibrary{arrows,automata} 
\usepackage{color}
\usepackage{graphicx}
\usepackage{mathrsfs}
\usepackage{url}
\usepackage{comment}
\usepackage{lipsum}
\usepackage{wasysym}%
\usepackage{cleveref}
\usepackage{thm-restate}
\usepackage[margin=1in]{geometry}

\usepackage{tikz}
\usetikzlibrary{arrows,automata} 
\usepackage{tikz-3dplot}
\usetikzlibrary{decorations.pathreplacing,calligraphy}
\usetikzlibrary{decorations.pathmorphing}

\tikzset{%
	pics/ellip/.style args={#1}{code={%
			\fill[rotate=60,#1] (0,0) ellipse (0.9pt and 1.1pt);
	}},
	pics/ellip/.default=CarnationPink, 
    > = stealth, 
	shorten > = 1pt, 
	auto,
	node distance = 3cm, 
	every edge/.append style = {thick}, 
}
\tikzstyle{axis}=[thin]
\tikzstyle{state}=
[circle,draw=black,thick,fill=black,minimum size=1.5mm,inner sep=0pt]

\usepackage[sort]{natbib}
\bibpunct[, ]{(}{)}{,}{a}{}{,}%

\newtheorem{example}{Example}
\newtheorem{theorem}{Theorem}
\newtheorem{lemma}{Lemma}
\newtheorem{proposition}{Proposition}

\newtheorem{definition}{Definition}

\newcommand{\R}{\mathbb{R}}

\newcommand{\calS}{\mathcal{S}}

\newcommand{\calX}{\mathcal{X}}

\newcommand{\cost}{C}
\newcommand{\prior}{\mu^*}
\newcommand{\support}{A}

\newcommand{\NN}{\mathbb{N}}

\newcommand{\Matri}{\changed{M}}
\newcommand{\matri}{\changed{m}}

\newcommand{\classP}{\textsf{P}}
\newcommand{\classNP}{\textsf{NP}}

\renewcommand{\phi}{\varphi}
\renewcommand{\epsilon}{\varepsilon}
\renewcommand{\setminus}{\, \backslash \,}

\newcommand{\D}{\displaystyle}

\allowdisplaybreaks

\title{Public Signals in Network Congestion Games}

\author{Svenja M. Griesbach\thanks{Centro de Modelamiento Matemático (CNRS IRL2807), Universidad de Chile, Chile. \texttt{sgriesbach@cmm.uchile.cl}} \and Martin Hoefer\thanks{Department of Computer Science, RWTH Aachen University, Germany. \texttt{mhoefer@cs.rwth-aachen.de}} \and Max Klimm\thanks{Institute of Mathematics, Technische Universit\"at Berlin, Germany. \texttt{klimm@math.tu-berlin.de}} \and Tim Koglin\thanks{Department of Computer Science, RWTH Aachen University, Germany. \texttt{koglin@algo.rwth-aachen.de}}}%

\date{}

\newcommand{\changed}[1]{{\color{black} #1}}

\begin{document}
    
\maketitle
		
\begin{abstract}
	Travel times in road networks are subject to stochastic uncertainty resulting from various parameters. A benevolent mobility service provider observing the actual travel times from data may use their informational advantage to steer the traffic equilibrium in a favorable direction by a task known as signaling or Bayesian persuasion. 
	Previous work has shown that the underlying signaling problem can be NP-hard to approximate within any non-trivial bounds, even for static Wardrop flows with affine cost functions with stochastic offsets. In contrast, we show that in this case, the signaling problem is easy for many networks. First, we tightly characterize the class of single-commodity networks, in which full information revelation is always an optimal signaling strategy. Second, we construct a reduction from optimal signaling to computing an optimal collection of support vectors for the Wardrop equilibrium. For two states, this allows us to efficiently compute an optimal signaling scheme whenever the number of different supports is bounded by a polynomial in the input size. Using a cell decomposition technique, we extend the approach to a polynomial-time algorithm for multi-commodity parallel edge networks with a constant number of commodities, even when we have a constant number of different states of nature.
\end{abstract}

\clearpage

	\section{Introduction}
		
	The selfish behavior of traffic participants degrades the performance of traffic networks. The theory of congestion games - within its long research history at the intersection of computer science, economics, and operations research - provides a rich set of tools to study this effect qualitatively and quantitatively. 
	
	Most notably, the substantial set of results includes characterizations and quantitative bounds on the inefficiency of equilibria. This inefficiency leads to higher overall travel times and, as a consequence, to higher greenhouse gas emissions, and a general loss of social welfare, compared to an optimal allocation of traffic flow in the network. Moreover, it can result in paradoxical effects of equilibrium behavior when changing the travel times or the demand in the network. 
    
	As a remedy for the inefficiency of equilibrium flows, various measures have been discussed in the literature, most prominently road tolls \citep[e.g.][]{Colini-Baldeschi18,FleischerJM04,KarakostasK04,HoeferOS08,HarksKKM15,PaccagnanCFM21,ColeDR06,JelinekKS14} and - to a lesser extent - network design techniques~\citep[e.g.][]{GairingHK17,Roughgarden06}. While in principle, some of these measures seem very attractive, the implementation of road toll stations, satellite-based methods to implement road pricing schemes, or the construction of new road infrastructure requires substantial and long-term investments and puts an additional burden on the environment. Consequently, it would be much more beneficial to be able to improve the traffic situation without the need for large investments.
	
	In this paper, we consider a largely untapped potential of network improvement that is rooted in the inherent uncertainty of travel times. In practice, the travel times are not deterministic and are subject to stochastic uncertainty (see, e.g., \cite{LianeasNM19} and further literature in \Cref{ssec:related-work}) due to various influences such as weather conditions, occurrences of road works, or traffic accidents. In these scenarios, mobility services like TomTom, Waze, or Google Maps have an informational advantage over a single network agent as they learn traffic conditions from data. 
	Indeed, with the omnipresence of navigation devices, mobility services are becoming increasingly aware of their significant impact on emerging traffic patterns~\citep{SpectrumIEEE19}. A benevolent mobility service may use its informational advantage to steer the traffic equilibrium in a favorable direction. The resulting optimization problem for the mobility service can be cast as a task commonly referred to as \emph{signaling} or \emph{Bayesian persuasion}. 
	
	In the Bayesian persuasion problem considered in this work, there is a set of different \emph{states of nature}, along with a prior distribution over these states. The state determines the exact cost functions for every connection in the network. Further, in this setting, the mobility service acts as a \emph{principal} who learns the realization of the state of nature by monitoring and aggregating available real-time data on traffic flows, weather forecasts, road work information, etc.\ In contrast, individual \emph{agents} only know the prior, but they rely on the mobility service for detailed information about the state of the traffic network. To exploit the informational advantage, the principal commits upfront to a \emph{signaling scheme}, i.e., a distribution over (abstract) signals for each possible state of nature. Examples of such signals are, e.g., route recommendations, traffic reports, weather updates, etc.\ Upon learning the realization of the state, the principal sends a (public) signal chosen according to the scheme to all agents. After having observed the signal, each agent performs a Bayesian update to adjust its belief about the realized state and the corresponding costs in the network and then chooses a route based on this updated belief. The goal of the principal is to choose the signaling scheme such that the expected cost of the emerging traffic equilibrium is minimized. For a more formal description of the model, we refer to Section~\ref{sec:preliminaries}.

    Every signal is an abstract message. As such, the sender can decide what information to communicate to the agents.
    For example, the signal may contain the corresponding equilibrium flow.
    In addition to sending the equilibrium flow as a \emph{public} signal to every agent, the sender can also send \emph{private} route recommendations to each agent that, combined, yield the equilibrium flow.
    Since under the conditions we impose in this paper, the Wardrop equilibrium is unique, the agents minimize their expected travel times by following these private route recommendations.
    This lifts the burden from the agents to compute an equilibrium upon receiving a signal.
    Interestingly, a combination of public signals (e.g., information about traffic delays) and private route recommendations  is also what is commonly issued by traffic service providers.
	
	This signaling problem is intrinsically difficult. In a seminal work, \cite{BhaskarCKS16} show that the problem of computing the optimal scheme cannot be approximated in polynomial time by a factor of $4/3 - \epsilon$ for every $\epsilon > 0$ unless \classP\ = \classNP. The hardness result applies to single-commodity networks (i.e., all agents travel from the same source to the same target) where edge costs are \emph{affine with state-based offsets} (i.e., the cost of an edge $e$ in state $\theta$ is of the form $c_e^{\theta}(x) = a_e \, x + b_e^{\theta}$ with $a_e>0$ and $b_e^\theta\geq 0$).
	A $4/3$-approximation, however, is a trivial consequence of the well-known bound of $4/3$ on the price of anarchy for affine costs~\citep{RoughgardenT02}. 
	
	The reduction of \citeauthor{BhaskarCKS16} yields strong bounds on the performance of polynomial algorithms on \emph{worst-case instances}. Yet, it leaves open the question of whether positive results can be obtained under stronger assumptions on the network structure and the number of states. From a practical point of view, it is of high interest to design algorithms that may have super-polynomial runtime on worst-case instances but can solve realistic instances within a reasonable time limit. These are the questions studied in this paper.

\subsection{Our Contribution}
We provide a thorough investigation of the power of information in congested networks.
More concretely, we study networks where the edge costs are affine with state-based offsets of the form $c_e^{\theta}(x) = a_e \, x + b_e^{\theta}$ with $a_e>0$ and $b_e^\theta\geq 0$, as in \cite{BhaskarCKS16}.
	As a prerequisite of our further results, we first show that the cost of the induced Wardrop equilibrium is a piecewise linear function of the agents' belief (\Cref{lem:piecewise-linear}). Specifically, we show that the function is affine as long as the support of the induced Wardrop equilibrium does not change. For the proof of this result, we use properties of the linear systems that characterize Wardrop equilibria for given supports that can be expressed via weighted graph Laplacians of appropriate subgraphs.
	
	We then study under which circumstances \emph{full information revelation} is an optimal signaling scheme. In this context, full information revelation means that the principal sends a distinct signal for each state so that the agents always know the realized state.
	Our first result (\Cref{thrm:sepa-full-opt-iff}) is a complete characterization of the single-commodity networks for which full information revelation is always an optimal signaling scheme.
	Specifically, we show that full information revelation is always optimal if and only if the network is series-parallel.
	Building on \Cref{lem:piecewise-linear}, we show that for a series-parallel network, the cost of the induced Wardrop equilibrium is not only piecewise linear, but also  concave, implying that full information revelation is an optimal signaling scheme.
	This characterization is tight in the sense that for every non-series-parallel network, there are cost functions such that full information revelation is not optimal.
	
	In \Cref{sec:LPs}, we tackle the computation of optimal signaling schemes in general multi-commodity networks. 
	To this end, we show that when a set of support vectors is given one can compute the optimal signaling scheme that is restricted to inducing Wardrop equilibria using only those supports (\Cref{thm:bigLP}).
	Hence, the hardness of finding good signaling schemes in~\cite{BhaskarCKS16} is caused by the difficulty of good \emph{support selection} for the emerging Wardrop equilibria.
	For games with two states, we also show that all supports of Wardrop equilibria can be enumerated in output-polynomial time (\Cref{prop:polyTwo}). Unfortunately, we show that even for games with two states, there can be exponentially many supports that arise in a Wardrop equilibrium for some $\mu \in \Delta(\Theta)$ (\Cref{thm:expoential-supports}).
    For parallel-edge networks, we exploit the ordinal structure of support sets and apply a cell decomposition technique to bound (and compute) the set of supports for all $\mu \in \Delta(\Theta)$. Our approach results in a polynomial-time algorithm to enumerate all supports whenever we have a constant number of commodities and a constant number of states. Hence, under these conditions, we can compute an optimal signaling scheme in polynomial time (\Cref{thm:parallel-poly-time}).

    While the number of supports can be exponential in general networks, we provide in Section~\ref{sec:compStudy} experimental evidence from actual road network data retrieved from an online library \citep{CSData22} that suggests that the number of different supports of the Wardrop equilibrium is indeed \emph{relatively small} (i.e., smaller than $15$ in all tested instances). Hence, our approach to compute optimal signals is viable for real-world networks. For the six network instances tested, we also provide further insights on the form of the resulting Wardrop equilibria as well as the benefit of signaling to reduce social cost. It turns out that full information is either the optimal or a near-optimal signaling scheme in all tested instances. 
    Hence, our results offer an appealing conclusion for mobility services: Full information revelation is simple as well as highly desirable from the perspective of the individual agent, and it also represents a \emph{near-optimal} signaling strategy in many real-world networks.

We believe that affine costs with stochastic offsets are a reasonable working assumption that enables us to understand the key aspects of relevant applications. Our main theoretical insights and analytical tools \emph{do not immediately extend} beyond affine costs with state-based offsets.
In particular, our main results build on the structural property, that the expected cost of the equilibrium as a function of the posterior belief is \emph{piecewise linear}.
As we show in \Cref{sec:discussion}, this condition does not continue to hold, even in a single-commodity network with two parallel links and affine costs with state-based slope (\Cref{ex:counter-stochastic-slope}) or quadratic cost with state-based offset (\Cref{ex:counter-monomials}).
    Piecewise linearity of the expected cost enables the construction of our linear programs to compute an optimal scheme for a bounded number of supports. As a consequence, this implies that an optimal scheme always involves probabilities that are rational numbers of bounded precision. Indeed, there are simple examples with affine costs and state-based slopes where the optimal signaling scheme involves irrational signaling probabilities (Example~\ref{ex:counter-stochastic-slope}).

    \subsection{Related Work}
    \label{ssec:related-work}
    
    \paragraph{General Bayesian Persuasion.}
	The idea of using informational advantages to enforce favorable outcomes in games goes back to the early work of \cite{AumannM66}; the area started to receive substantial interest after the seminal work of \cite{KamenicaG11}. Characterizing and computing good signaling schemes is a very active area of research \citep{DughmiX21, BadanidiyuruBX18, EmekFGLT14}, involving diverse aspects such as, e.g., limited signals~\citep{DughmiKQ16,GradwohlHHS21}, multiple receivers and private signals~\citep{Rubinstein17,ArieliB19,Xu20}, online optimization~\citep{HahnHS20IJCAI,HahnHS20}, and learning~\citep{CastiglioniCMG20,CastiglioniMCG21,ZuIX21}. 
	
	\paragraph{Congestion Games and Stochasticity.}
	It has been recognized that in practice travel times are non-deterministic \citep{LiuBRM02,MirchandaniS87}.
	This motivates the study of non-atomic congestion games with stochastic travel times for which equilibrium characterizations for risk-averse agents are obtained \citep{CominettiT16,OrdonezSM10,NikolovaM14,Nie11}.
	Bounds on the efficiency loss due to risk aversion are studied by \citet{LianeasNM19}.
	Further related is the stochastic agent equilibrium concept introduced by \citet{Dial71} where travel times are deterministic but are perceived stochastically by the agents.
	
	Another source of uncertainty in non-atomic congestion games is the demand.
	The dependency of the price of anarchy on the demand is studied empirically by \citet{YounGJ08,OHareCW16} and analytically by \citet{Colini-Baldeschi19,Colini-Baldeschi20,CominettiDS24,WuMCX21}.
	\citet{WangDC14} study the price anarchy for stochastic demands. 
	More generally, the sensitivity of Wardrop equilibria has been analyzed by \citet{EnglertFO10,KlimmW22,Patriksson04}.
	\citet{WuMRX23} quantify the efficiency loss due to selfish behavior in atomic congestion games with growing demand.
	Further related are atomic congestion games where agents participate with a certain probability \citep{GairingMT08,CominettiSSSM19,AngelidakisFL13,AshlagiMT06,MeirTBK12}. These games are known to admit multiple Nash equilibria, so that one faces the issue of equilibrium selection for the design of the optimal public signaling scheme. However, \citet{CominettiSSM23} show that Nash equilibria in a sequence of weighted congestion games with $n$ players converge to a Wardrop equilibrium of a limiting non-atomic game for $n \to \infty$ under mild and natural conditions.
		
	\paragraph{Bayesian Persuasion in Congestion Games.}
	\cite{BhaskarCKS16} consider the same model as this paper (i.e., \emph{non-atomic} congestion games with affine cost functions and state-based offsets) and show that it is $\mathsf{NP}$-hard to compute a signaling policy that approximates the total expected travel time better than a factor of $4/3 - \epsilon$ for any $\epsilon > 0$.
	\cite{Das17} examine two concrete networks with affine costs, one with state-based offsets and one with state-based slope, and compute the optimal signaling scheme.
    \citet{NachbarX21} study among others general state-based cost functions and the connection between signaling schemes and the price of anarchy.
    \cite{Massicot19} provide a complete characterization of the optimal public signaling policy with affine costs for the single-commodity case on a network consisting of two parallel edges. 
    \cite{Vasserman15} consider parallel edges with affine costs. The cost functions are permuted randomly and bound the performance of private signals.
    \cite{Zhu22} consider a model where only a fixed fraction of the agents receives private or public signals.
    Their main focus is on a parallel edges network with two edges.
    For polynomial latency functions, they derive a connection to the generalized problem of moments and associated semidefinite programming techniques to approximate optimal signaling policies. They establish a formulation of the optimization problem for private signaling policies.
    In contrast, our focus is on public signals and affine costs which allows us to obtain exact algorithms for the signaling problem.
    \cite{GriesbachHKK24} study the impact of information design in networks with affine cost functions and uncertainty in the demand.

    \cite{CastiglioniCM021} study a related model for \emph{atomic} congestion games, where the players have to commit to following the mediator's recommendation before receiving the signal. Their positive results are contrasted by recent work of \cite{ZhouNX22} for singleton games in the standard signaling model.   
    \citet{AcemogluMMO18} consider a setting where players have only partial knowledge concerning the available edges on the graph and give a complete characterization of the graph classes for which players cannot obtain higher cost by gaining additional information.
    \citet{WuAO21} study the Bayesian Wardrop equilibria arising when commodities receive multiple signals from different information systems.
    
    \citet{Arnott91} explore the provision of information in a \emph{dynamic} model where players have preferences over arrival times.
    \citet{GriesbachKWZ23} consider a dynamic flow model and use Bayesian persuasion to optimize throughput and makespan.

An upshot of the work of \citet{AcemogluMMO18} is the connection between information design for congested networks and the Braess' paradox \citep{Braess68}. In the latter, the addition of a zero-cost edge increases the travel time of all traffic participants. Further results in this direction include \cite{Milchtaich06}  who characterized the networks where this paradox can occur, \cite{FujishigeGHPZ17} who showed that the paradox does not appear when the strategies are a matroid, \cite{Roughgarden06} who studied a related computational problem, and \cite{ValiantR10} who computed the probability of the occurrence of the paradox in random networks.

Further, as noted by \cite{NachbarX21}, Bayesian persuasion is intrinsically connected to the price of anarchy since information design cannot decrease the total travel time by a factor that is larger than the price of anarchy. Tight bounds on the price of anarchy for various sets of cost functions have been established both for non-atomic congestion games \citep{CorreaSM08,Roughgarden03,RoughgardenT02,Roughgarden12CACM} and their atomic counterparts~\citep{AlandDGMS11,AwerbuchAE13,ChristodoulouGGS19,ChristodoulouK05}. Many of these bounds depend heavily on the growth rate of the cost functions that measure delays along connections in the networks. In contrast, the structure of networks resulting in worst-case inefficiency is often rather generic.

	\section{Model and Preliminaries}
	\label{sec:preliminaries}

	Let $\Theta$ be a finite set of \emph{states of nature} and let $G=(V,E)$ be a directed graph with edge set $E$. For every edge~$e \in E$ and every state $\theta \in \Theta$ , there is a \emph{cost function} $c_{e}^\theta : \R_{\geq 0} \to \R_{\geq 0}$ of the form
	\[
		c_e^{\theta}(x) = a_e \, x + b_e^{\theta} \quad \text{ with $a_e \in \mathbb{R}_{>0}$ and $b_e^{\theta} \in \mathbb{R}_{\geq 0}$.}
	\]
	Further, there is a set $R = \{1,\ldots,r\}$ of \emph{commodities}.	
	Every commodity  $i \in R$ corresponds to a continuum of agents of total volume $d_i > 0$, and is associated with the interval $[0,d_i]$.
    For every commodity $i \in R$, there is a designated source vertex $s_i \in V$ and a designated destination vertex~$t_i\in V$.
    Let $\mathcal{S}_i\subseteq 2^E$ be the set of all $s_i$--$t_i$ paths. Then, a \emph{path-flow} is a distribution of the agents of each commodity $i\in R$ over $\calS_i$ and is represented by a vector $x = (x_{i,S})_{i \in R, S \subseteq E}$ satisfying the three properties:
    \begin{enumerate}
    	\item $\sum_{S \subseteq E} x_{i,S} = d_i$ for all $i\in R$,
    	\item $x_{i,S} \geq 0$ for all $i \in R$, $S \subseteq E$, and
    	\item $x_{i,S} = 0$ for all $i \in R$, $S \notin \calS_i$.
    \end{enumerate}
    Let $\calX$ denote the set of those vectors. Every path-flow $x \in \calX$ induces a \emph{load} on every edge~$e \in E$ given by
	\[
	x_e = \sum_{i \in R} \sum_{S \subseteq E : e \in S} x_{i,S}.
	\]
	
	Let $\Delta(\Theta)$ be the set of probability distributions over the states of nature $\Theta$, i.e., $\Delta(\Theta)=\bigl\{\mu\in [0,1]^{|\Theta|} : \sum_{\theta\in\Theta} \mu_\theta =1 \bigr\}$.
	An element $\mu\in\Delta(\Theta)$ is called a \emph{belief} and induces \emph{believed cost} for each edge~$e$ via 
	\[
		c_e(x_e \mid \mu) = \sum_{\theta \in \Theta} \mu_{\theta}\, c_e^\theta(x_e).
	\]
	A path-flow $x \in \mathcal{X}$ is a \emph{Wardrop equilibrium} for belief $\mu$ if all agents only use paths that are minimal with respect to the believed costs, i.e., if
	\begin{align*}
	\sum_{e\in S}c_e(x_e\mid \mu)
	\leq \sum_{e\in S'} c_e(x_e\mid \mu) \quad \text{ for all } i\in R \text{ and } S,S'\in \mathcal{S}_i \text{ with } x_{i,S}>0.
	\end{align*}
	
	The existence of Wardrop equilibria for non-atomic games with a \emph{single} state \citep[cf.][]{BeckmannMW56} carries over to arbitrary beliefs over multiple states in the following sense.

	\begin{proposition}
		\label{pro:beckmann}
		Given any belief $\mu \in \Delta(\Theta)$, a path-flow $x \in \mathcal{X}$ is a Wardrop equilibrium if and only if 
		\begin{align*}
			x \in \arg\min \Biggl\{\sum_{e \in E} \int_{0}^{y_e} c_{e}(t \mid \mu)\,\mathrm{d}t \;:\; y \in \mathcal{X}\Biggr\}.
		\end{align*}	
	\end{proposition}

    Since all cost functions $c_e^\theta$ are increasing (i.e., $a_e > 0$ for all $e \in E$), so are their convex combinations $c_e(x \mid \mu)$. Hence, the optimization problem in Proposition~\ref{pro:beckmann} is strictly convex and has a unique solution, i.e., the Wardrop equilibrium is unique. In the following, we denote by $x^*(\mu)$ the unique Wardrop equilibrium for a belief $\mu \in \Delta(\Theta)$. Let further
	\begin{align*}
		\cost(x \mid \mu) = \sum_{e \in E} x_e c_e(x_e \mid \mu)
	\end{align*}
	denote the total cost of a path-flow $x$ for $\mu \in \Delta(\Theta)$. For the Wardrop equilibrium $x^*(\mu)$ for $\mu \in \Delta(\Theta)$, we use the short notation $\cost(\mu) = \cost(x^*(\mu) \mid \mu)$.
	The thus defined function $C:\Delta(\Theta)\rightarrow \R$ assigns each belief $\mu$ the cost of the resulting Wardrop equilibrium.
    Different beliefs may lead to different Wardrop equilibria and, hence, to different costs.
	This paper is concerned with the question how a benevolent principal, who knows the true realization $\theta\in \Theta$, can (partially) reveal this information to induce equilibria with low cost.
	This so-called \emph{signaling problem} is formalized in the following way.
	
	Without having received any information from the principal, the agents act upon a common prior belief $\mu^* \in \Delta(\Theta)$ that corresponds to the true probabilities with which the corresponding states are realized. 
	The principal conveys information to the agents in the form of public signals. To this end, the principal has access to a finite set of signals $\Sigma$.
	A \emph{signaling scheme} is given by a matrix $\Phi = (\phi_{\theta,\sigma})_{\theta \in \Theta, \sigma \in \Sigma}$ such that $\phi_{\theta,\sigma} \in [0,1]$ is the probability that state $\theta$ is realized and signal~$\sigma$ is issued.
	Since the prior belief corresponds to the actual realization probabilities, we have the equation $\sum_{\sigma \in \Sigma} \phi_{\theta,\sigma} = \prior_{\theta}$ for each $\theta \in \Theta$. 	
	Let further $\phi_\sigma = \sum_{\theta \in \Theta} \phi_{\theta,\sigma}$ be the total probability of issuing signal $\sigma$.
	
	We proceed to explain the timeline of the signaling problem. First, the principal commits to a signaling scheme $\Phi$ and communicates this to all players so 
    that the prior $\prior$ and the signaling scheme $\Phi$ are public knowledge. Then the state of nature $\theta$ is realized. The principal perceives the realized state $\theta$ and sends public signal $\sigma$ with probability $\phi_{\theta,\sigma} / \mu^*_{\theta}$.
	All agents receive the same signal~$\sigma$ and update their belief about the states of nature by a Bayesian update, i.e., their posterior belief is $\mu_{\sigma} = (\mu_{\theta,\sigma})_{\theta \in \Theta}\in \Delta(\Theta)$ defined by $\mu_{\theta,\sigma} = \phi_{\theta,\sigma} / \phi_{\sigma}$ for all $\theta \in \Theta$.
	Afterward, the corresponding Wardrop equilibrium $x^*(\mu_\sigma)$ emerges and results in cost of $C(\mu_\sigma)$. 
	Hence, every signaling scheme $\Phi$ induces posterior beliefs $\mu_\sigma$ for all $\sigma \in \Sigma$ and, thus, yields an overall cost of
	\begin{align}
	\label{eq:c} 
		\cost(\Phi) = \sum_{\sigma \in \Sigma} \phi_{\sigma} \, \cost(\mu_\sigma).
	\end{align}
    The goal of the principal is to choose the signaling scheme $\Phi$ that minimizes \eqref{eq:c}.

    Mathematically, every signaling scheme $\Phi$ yields a convex decomposition of the prior belief $\mu^*=\sum_{\sigma\in\Sigma}\phi_\sigma \mu_\sigma$. 
     Conversely, it can be shown 
     \citep{AumannM95,KamenicaG11} that every convex decomposition of the prior belief can be realized by a signaling scheme.

    Further, it is a direct consequence of Caratheodory's Theorem that $|\Theta|$ many different signals suffice to obtain the desired convex decomposition.
    As a consequence, the problem of computing a signaling scheme minimizing \eqref{eq:c} can be phrased as
    \begin{align*}
    	\inf \bigg\{ \sum_{\sigma\in \Sigma} \phi_\sigma C(\mu_\sigma) : |\Sigma|\leq |\Theta|, \phi_\sigma\in[0,1],\mu_\sigma\in\Delta(\Theta) \text{ f.a.}  \sigma\in \Sigma \text{ s.t. } \sum_{\sigma\in\Sigma}\phi_\sigma\mu_\sigma=\mu^* , \sum_{\sigma \in \Sigma} \varphi_\sigma = 1 \bigg\}.
    \end{align*}

The cost of an optimal scheme as a function of $\mu^*$ corresponds to the lower convex envelope of the cost function $\cost$ \citep{KamenicaG11}. Therefore, in the case $|\Theta| = 2$, if the signaling scheme $\Phi$ has a cost that can be expressed as a mixture of $\cost((1,0)$ and $\cost((0,1))$, then full information revelation optimal. This means that optimality of full information revelation readily follows if $C$ is a concave function.
    
We illustrate these concepts with the following two examples. The first example presents an instance with two states for which full information revelation is optimal.

\begin{figure}[t]
\scriptsize
\begin{subfigure}[b]{0.3\textwidth}
\begin{center}
\begin{tikzpicture}[scale = 0.7]
\node[state,draw=Green,fill=Green,label=left:{$s$}] (s) at (0,0) [circle] {}; 
\node[state,draw=Green,fill=Green,label=right:{$t$}] (t) at (4,0) [circle] {}; 
				
\path[->,Green] (s) edge[bend left=60]  node[above,yshift=0mm] {\textcolor{black}{$2x+5$}} (t); 
\path[->,Green] (s) edge[bend right=60]  node[below,yshift=0mm] {\textcolor{black}{$2x$}} (t); 
\path[->,Green] (s) edge node[below,yshift=0mm] {\textcolor{black}{$3$}} (t);   
\end{tikzpicture}
\end{center}
\caption{\label{fig:concave1}}
\end{subfigure}
\begin{subfigure}[b]{0.3\textwidth}
\begin{center}
\begin{tikzpicture}[scale = 0.7]
\node[state,draw=Red,fill=Red,label=left:{$s$}] (s) at (0,0) {}; 
\node[state,draw=Red,fill=Red,label=right:{$t$}] (t) at (4,0) {}; 
				
\path[->,Red] (s) edge[bend left=60]  node[above,yshift=0mm] {\textcolor{black}{$2x$}} (t); 
\path[->,Red] (s) edge[bend right=60]  node[below,yshift=0mm] {\textcolor{black}{$2x+4$}} (t); 
\path[->,Red] (s) edge node[below,yshift=0mm] {\textcolor{black}{$3$}} (t);   	
\end{tikzpicture}
\end{center}
\caption{\label{fig:concave2}}
\end{subfigure}
\begin{subfigure}[b]{0.39\textwidth}
\begin{center}
\begin{tikzpicture}[scale=0.7]
\draw[axis,->] (0,-0.2) to (0,3.5) node[above] {$\cost(\mu)$}; 
\draw[axis,->] (-0.2,0) to (4.5,0) node[right] {$\mu_{\theta_2}$};
				
\draw[axis] (4,0.2) -- (4,-0.2) node[below] {$1$};
\draw[axis] (3,0.2) -- (3,-0.2);
\draw[axis] (2,0.2) -- (2,-0.2);
\draw[axis] (1,0.2) -- (1,-0.2);
\draw[axis] (0,0.2) -- (0,-0.2) node[below] {$0$};

\draw[axis] (0.2,0) -- (-0.2,0) node[left] {$0$};
\draw[axis] (0.2,1) -- (-0.2,1);
\draw[axis] (0.2,2) -- (-0.2,2);
\draw[axis] (0.2,3) -- (-0.2,3) node[left] {$3$};

\draw[thick,YellowOrange] (0,2) -- (4,2);
    
\draw[ultra thick,CornflowerBlue!50!MidnightBlue] (0,2) -- (1,3) -- (3.2,3) -- (4,2);
\fill[Green] (0,2) circle (3.0pt);
\fill[Green] (0,0) circle (3.0pt);
\draw[thick,Green] (0,2) -- (0,0);

\fill[Red] (4,2) circle (3.0pt);
\fill[Red] (4,0) circle (3.0pt);
\draw[thick,Red] (4,2) -- (4,0);
\end{tikzpicture}
\end{center}
\caption{}
\label{fig:concave3}
\end{subfigure}
\caption{Instance in Example~\ref{ex:concave}: (a) cost functions for state $\theta_1$; (b) cost functions for state $\theta_2$; (c) cost $\cost(\mu)$ of the Wardrop equilibrium as a function of the belief, described by parameter $\mu_{\theta_2} = 1-\mu_{\theta_1} \in [0,1]$.}
\label{fig:concaveExample}
\end{figure}
	
	\begin{example}\label{ex:concave} \rm
		Consider a single-commodity network with two vertices $V=\{s,t\}$ and three parallel edges $E = \{e_1,e_2,e_3\}$. There are two states $\Theta=\{\theta_1,\theta_2\}$ and a single commodity with a volume of $d_1=1$. The cost functions $c_{e_i}^{\theta_1}$ and $c_{e_i}^{\theta_2}$ are given in Figure~\ref{fig:concave1} and \ref{fig:concave2}, respectively. 
		
		We analyze the Wardrop equilibrium for all distributions $\mu \in \Delta(\Theta)$, where $\Delta(\Theta)$ is interpreted as the unit interval for $\mu_{\theta_2} = 1 - \mu_{\theta_1} \in [0,1]$.
		For $\mu_{\theta_2} \in [0,1/4]$, only the lowest edge is used, since the total cost for a volume of $1$ is at most $3$, whereas both other edges have an offset of at least $3$. For $\mu_{\theta_2} \in [1/4, 2/5]$, only the two lower edges are used, since the upper edge has an offset of at least 3. For $\mu_{\theta_2} \in [2/5, 3/4]$ all three edges carry flow. Then, for $\mu_{\theta_2} \in [3/4, 4/5]$, only the two upper edges are used, since the offset of the lowest edge is at least $3$. Finally, for $\mu_{\theta_2} \in [4/5, 1]$, only the upper edge is used. \Cref{fig:concave3} shows the cost function $\cost(\mu)$ of the resulting Wardrop equilibrium for all $\mu \in \Delta(\Theta)$ in blue. The cost function $\cost(\mu)$ is piecewise linear and concave over $\Delta(\Theta)$. 
		It is easy to convince ourselves that full information revelation is an optimal signaling scheme. Indeed, assume there is a signaling scheme $\Phi$ that decomposes the prior $\mu_{\theta_2}^*$ into the two beliefs $\mu_l$ and $\mu_r$, such that $\mu_l\leq \mu_{\theta_2}^*\leq \mu_r$. Then,
		\begin{align*}
			 C(\Phi)=\frac{\mu_r-\mu_{\theta_2}^*}{\mu_r-\mu_l} \cost(\mu_l) + \frac{\mu_{\theta_2}^*-\mu_l}{\mu_r-\mu_l} \cost(\mu_r).
		\end{align*}
		However, by concavity, we have $C(\mu_l)\geq \mu_lC(1)+(1-\mu_l) C(0)$ and $C(\mu_r)\geq \mu_r C(1)+(1-\mu_r) C(0)$. This yields
		\begin{align*}
			 C(\Phi)
			 &\geq \frac{\mu_r-\mu_{\theta_2}^*}{\mu_r-\mu_l} \Big( \mu_lC(1)+(1-\mu_l) C(0) \Big) + \frac{\mu_{\theta_2}^*-\mu_l}{\mu_r-\mu_l} \Big(\mu_r C(1)+(1-\mu_r) C(0)\Big)\\
			 &\geq (1-\mu_{\theta_2}^*) \cost(0) + \mu_{\theta_2}^* \cost(1),
		\end{align*}
		where the latter corresponds to the cost of the full information revelation scheme, shown in orange in \Cref{fig:concave3}.
		\hfill $\blacksquare$
	\end{example}

In the second example, we consider an instance with three states for which the optimal signaling scheme does not reveal full information.
	
\begin{example}
 \label{ex:3d}

 Consider the single-commodity network with four vertices in Figures~\ref{fig:3d1}--\ref{fig:3d3}. There are three states $\theta_1$, $\theta_2$, and $\theta_3$. There is a single commodity with volume $d_1 = 1$. The cost functions for the three states are shown in Figures~\ref{fig:3d1}--\ref{fig:3d3}. Figure~\ref{fig:3d4} shows the cost of the Wardrop equilibrium as a function of the distribution $\mu$ where only $\mu_{\theta_2}, \mu_{\theta_3} \in [0,1]$ with $\mu_{\theta_2} + 
 \mu_{\theta_3} \leq 1$ are shown and $\mu_{\theta_1} = 1 - \mu_{\theta_2} - \mu_{\theta_3}$. Every linear segment of the cost function corresponds to a support of the Wardrop equilibrium of the underlying network. The optimal convex decomposition of the prior $\mu^* = (1/3,1/3,1/3)$ shown in Figure~\ref{fig:3d4} is via the red point in the front, the blue point on the right and the purple point in the back and is indicated by the orange hyperplane through these three points. It is worth noting that the purple point corresponds to a signal that is issued with positive probability both for state $\theta_2$ and $\theta_3$. Thus, the optimal signaling scheme is not full information revelation.  
 \hfill $\blacksquare$

\begin{figure}
\scriptsize
\newcommand{\zscale}{0.2}
\tdplotsetmaincoords{70}{110}
\begin{minipage}[b]{0.33\textwidth}
	\begin{subfigure}[b]{\textwidth}
\begin{center}
\begin{tikzpicture}[xscale=3,yscale=2]
\begin{scope}
\node[state,draw=Green,fill=Green] (v1) at (0,0.8) {};
\node[state,draw=Green,fill=Green,label=right:{$t$}] (t1) at (0.5,0.4) {};
\node[state,draw=Green,fill=Green,label=left:{$s$}] (s1) at (-0.5,0.4) {};
\node[state,draw=Green,fill=Green] (w1) at (0,0) {};

\draw[->,Green] (s1) to node[above left] {\textcolor{black}{$x$}} (v1); 
\draw[->,Green] (v1) to node[above right] {\textcolor{black}{$3$}} (t1); 
\draw[->,Green] (w1) to node[below right] {\textcolor{black}{$x+1$}} (t1); 
\draw[->,Green] (s1) to node[below left] {\textcolor{black}{$5$}} (w1); 
\draw[->,Green] (v1) to node[right] {\textcolor{black}{$1$}} (w1); 
\end{scope}
\end{tikzpicture}
\end{center}
\caption{}
\label{fig:3d1}
\end{subfigure}
\vspace{0.3cm}

\begin{subfigure}[b]{\textwidth}
\begin{center}
\begin{tikzpicture}[xscale=3,yscale=2]
\begin{scope}
\node[state,draw=Red,fill=Red] (v2) at (0,0.8) {};
\node[state,draw=Red,fill=Red,label=right:{$t$}] (t2) at (0.5,0.4) {};
\node[state,draw=Red,fill=Red,label=left:{$s$}] (s2) at (-0.5,0.4) {};
\node[state,draw=Red,fill=Red] (w2) at (0,0) {};

\draw[->,Red] (s2) to node[above left] {\textcolor{black}{$x+2$}} (v2); 
\draw[->,Red] (v2) to node[above right] {\textcolor{black}{$2$}} (t2); 
\draw[->,Red] (w2) to node[below right] {\textcolor{black}{$x+2$}} (t2); 
\draw[->,Red] (s2) to node[below left] {\textcolor{black}{$2$}} (w2); 
\draw[->,Red] (v2) to node[right] {\textcolor{black}{$2$}} (w2); 
\end{scope}
\end{tikzpicture}
\end{center}
\caption{}
\label{fig:3d2}
\end{subfigure}
\vspace{0.3cm}

\begin{subfigure}[b]{\textwidth}
\begin{center}
\begin{tikzpicture}[xscale=3,yscale=2]
\begin{scope}
\node[state,draw=MidnightBlue,fill=MidnightBlue] (v3) at (0,0.8) {};
\node[state,draw=MidnightBlue,fill=MidnightBlue,label=right:{$t$}] (t3) at (0.5,0.4) {};
\node[state,draw=MidnightBlue,fill=MidnightBlue,label=left:{$s$}] (s3) at (-0.5,0.4) {};
\node[state,draw=MidnightBlue,fill=MidnightBlue] (w3) at (0,0) {};

\draw[->,MidnightBlue] (s3) to node[above left] {\textcolor{black}{$x+3$}} (v3); 
\draw[->,MidnightBlue] (v3) to node[above right] {\textcolor{black}{$2$}} (t3); 
\draw[->,MidnightBlue] (w3) to node[below right] {\textcolor{black}{$x$}} (t3); 
\draw[->,MidnightBlue] (s3) to node[below left] {\textcolor{black}{$0$}} (w3); 
\draw[->,MidnightBlue] (v3) to node[right] {\textcolor{black}{$0$}} (w3); 
\end{scope}
\end{tikzpicture}
\end{center}
\caption{}
\label{fig:3d3}
\end{subfigure}
\end{minipage}
\begin{minipage}[b]{0.66\textwidth}
	\begin{subfigure}[b]{\textwidth}
\begin{center}
\begin{tikzpicture}[xscale=4,yscale=4,tdplot_main_coords]

\tikzstyle{face}=[rounded corners=0.1,line width=0.2pt, fill=teal!40, fill opacity = .5];

\draw[axis,->] (0,0,0) -- (1.2,0,0);
\node (label2) at (1.35,0,0) {$\mu_{\theta_2}$}; 
\draw[axis,->] (0,0,0) -- (0,1.2,0);
\node (label3) at (0,1.35,0) {$\mu_{\theta_3}$}; 
\draw[axis,->] (0,0,0) -- (0,0,\zscale*6);
\node (labelz) at (0,0,\zscale*6.6) {$C(\mu)$};

\draw[axis] (0,0.02,\zscale*0) -- (0,-0.02,\zscale*0) node[left] {$0$};
\draw[axis] (0,0.02,\zscale*1) -- (0,-0.02,\zscale*1) node[left] {$1$};
\draw[axis] (0,0.02,\zscale*2) -- (0,-0.02,\zscale*2) node[left] {};
\draw[axis] (0,0.02,\zscale*3) -- (0,-0.02,\zscale*3) node[left] {};
\draw[axis] (0,0.02,\zscale*4) -- (0,-0.02,\zscale*4) node[left] {};
\draw[axis] (0,0.02,\zscale*5) -- (0,-0.02,\zscale*5) node[left] {$5$};

\draw[thin] (0.5,0.02,0) -- (0.5,-0.02,0) node[left] {$1/2$};
\draw[thin] (1,0.02,0) -- (1,-0.02,0) node[left] {$1$};

\draw[thin] (0.02,0.5,0) -- (-0.02,0.5,0) node[above right] {$\!\!\!\!1/2$};
\draw[thin] (0.02,1,0) -- (-0.02,1,0) node[above right] {$1$};

\coordinate (A1) at (1/9,1/3, \zscale*43/9);
\coordinate (A2) at (0,3/7, \zscale*31/7);
\coordinate (O)  at (0, 0, \zscale*4);
\coordinate (A10) at (1/9,1/3, 0);
\coordinate (A20) at (0,3/7, 0);
\coordinate (O0)  at (0, 0, 0);

\coordinate (B1) at (10/27, 1/9, \zscale*124/27);
\coordinate (B2) at (1/3, 0, \zscale*13/3);
\coordinate (B3) at (2/3, 0, \zscale*14/3);
\coordinate (B10) at (10/27, 1/9, 0);
\coordinate (B20) at (1/3, 0, 0);
\coordinate (B30) at (2/3, 0, 0);

\coordinate (C1) at (4/5, 1/5, \zscale*21/5);
\coordinate (C2) at (0, 1, \zscale*1);
\coordinate (C3) at (0, 4/7, \zscale*25/7);
\coordinate (C4) at (4/27, 4/9, \zscale*109/27);
\coordinate (C10) at (4/5, 1/5, 0);
\coordinate (C20) at (0, 1, 0);
\coordinate (C30) at (0, 4/7, 0);
\coordinate (C40) at (4/27, 4/9, 0);

\coordinate (D1) at (1, 0, \zscale*9/2);
\coordinate (D2) at (C1);
\coordinate (D3) at (C4);
\coordinate (D4) at (B1);
\coordinate (D5) at (B3);
\coordinate (D10) at (1, 0, 0);
\coordinate (D20) at (C10);
\coordinate (D30) at (C40);
\coordinate (D40) at (B10);
\coordinate (D50) at (B30);

\coordinate (E1) at (O);
\coordinate (E2) at (A1);
\coordinate (E3) at (B1);
\coordinate (E4) at (B2);
\coordinate (E10) at (O0);
\coordinate (E20) at (A10);
\coordinate (E30) at (B10);
\coordinate (E40) at (B20);

\coordinate (F1) at (A1);
\coordinate (F2) at (C4);
\coordinate (F3) at (C3);
\coordinate (F4) at (A2);
\coordinate (F10) at (A10);
\coordinate (F20) at (C40);
\coordinate (F30) at (C30);
\coordinate (F40) at (A20);

\coordinate (G1) at (A1);
\coordinate (G2) at (C4);
\coordinate (G3) at (B1);
\coordinate (G10) at (A10);
\coordinate (G20) at (C40);
\coordinate (G30) at (B10);

\draw[face,fill=red,fill opacity=0,gray] (O0) -- (A10) -- (A20) -- cycle; 
\draw[face,fill=blue,fill opacity=0,gray] (B10) -- (B20) -- (B30) -- cycle; 
\draw[face,fill=pink,fill opacity=0,gray] (C10) -- (C20) -- (C30) -- (C40) -- cycle; 
\draw[face,fill=brown,fill opacity=0,gray] (D10) -- (D20) -- (D30) -- (D40) -- (D50) -- cycle; 
\draw[face,fill=orange,fill opacity=0,gray] (E10) -- (E20) -- (E30) -- (E40) -- cycle; 
\draw[face,fill=green,fill opacity=0,gray] (F10) -- (F20) -- (F30) -- (F40) -- cycle; 
\draw[face,fill=teal,fill opacity=0,gray] (G10) -- (G20) -- (G30) -- cycle; 
\draw[face,fill=white, fill opacity=0, darkgray!60!gray] (O0) -- (D10) -- (C20) -- cycle;

\draw[thick,CarnationPink] (1/3,1/3,\zscale*25/6) -- (1/3,1/3,0);
\draw[thick,DarkOrchid] (B2) -- (1/3, 0, 0);
\draw[thick,Green] (0,0,\zscale*4) -- (0, 0, 0);
\draw[thick,MidnightBlue] (C2) -- (0, 1, 0);
\draw[thick,Red] (D1) -- (1, 0, 0);

\pic at (1/3,1/3,\zscale*3.25) {ellip};

\draw[draw=none,face,fill=YellowOrange!80!Yellow] (B2) -- (C2) -- (D1) -- cycle;

\draw[face,gray,fill=CornflowerBlue!70!white] (D1) -- (D2) -- (D3) -- (D4) -- (D5) -- cycle; 
\draw[face,gray,fill=CornflowerBlue!40!white] (C1) -- (C2) -- (C3) -- (C4) -- cycle; 

\draw[face,gray,fill=CornflowerBlue!30!MidnightBlue] (B1) -- (B2) -- (B3) -- cycle; 
\draw[face,gray,fill=CornflowerBlue!80!MidnightBlue] (G1) -- (G2) -- (G3) -- cycle; 
\draw[face,gray,fill=CornflowerBlue!70!white] (F1) -- (F2) -- (F3) -- (F4) -- cycle; 

\draw[face,gray,fill=CornflowerBlue!30!MidnightBlue] (E1) -- (E2) -- (E3) -- (E4) -- cycle; 
\draw[face,gray,fill=CornflowerBlue!10!MidnightBlue] (O) -- (A1) -- (A2) -- cycle; 

\fill[Green] (0,0,\zscale*4) circle (0.4pt);
\fill[Green] (0,0,0) circle (0.4pt);

\fill[MidnightBlue] (C2) circle (0.4pt);
\fill[MidnightBlue] (0,1,0) circle (0.4pt);
\fill[Red] (D1) circle (0.4pt);
\fill[Red] (1,0,0) circle (0.4pt);

\fill[DarkOrchid] (B2) circle (0.4pt);
\fill[DarkOrchid] (1/3,0,0) circle (0.4pt);

\fill[CarnationPink] (1/3,1/3,\zscale*25/6) circle (0.4pt);
\fill[CarnationPink] (1/3,1/3,0) circle (0.4pt);
\end{tikzpicture}
\end{center}
\caption{}
\label{fig:3d4}
\end{subfigure}
\end{minipage}

\caption{Instance in Example~\ref{ex:3d}: (a) cost functions for state $\theta_1$; (b) cost functions for state $\theta_2$; (c) cost functions for state $\theta_3$; (d) cost $C(\mu)$ of a Wardrop equilibrium as a function of the belief, described by the parameters $\mu_{\theta_2}$ and $\mu_{\theta_3}$ with $\mu_{\theta_2} + \mu_{\theta_3} \leq 1$.}
\label{fig:3d}
\end{figure}
\end{example}

\section{Structural Properties}
\label{sec:structural}

In this section, we provide structural properties of the cost of the induced Wardrop equilibrium as a function of the belief. We first prove that for general single-commodity network congestion games, the cost of the unique Wardrop equilibrium with respect to $c_e(\cdot \mid \mu)$ is piecewise linear in $\mu \in \Delta(\Theta)$.
For the proof, we introduce the following notation. For a fixed source vertex $s$, we define the set of edges
\begin{align*}
\mathcal{A} = \{A \subseteq E : G = (V,A) \text{ is connected and contains an $s$--$v$ path for all $v \in V$}\} \, .
\end{align*}
Let $x$ be an $s$--$t$ flow. For $v \in V$, let $\psi_v$ be the cost of a shortest path with respect to $c_e(x_e \mid \mu)$ from $s$ to $v$. We call an edge $e = (v,w)$ \emph{active} in $x$ if $\psi_w - \psi_v = c_e(x_e \mid \mu)$. Let $A(x)$ be the set of active edges for a flow $x$. For every flow $x$, $A(x)$ is connected and every vertex $v$ is reached by a path of active edges from $s$; thus $A(x) \in \mathcal{A}$.  

\begin{definition}[Support]
Let $x^* : \Delta(\Theta) \to \mathbb{R}_{\geq 0}^E$ be the unique Wardrop equilibrium with respect to $c_e(\cdot \mid \mu)$ with costs $C : \Delta(\Theta) \to \mathbb{R}_{\geq 0}$ defined as $\sum_{e \in E} x_e^* c_e(x_e^* \mid \mu)$. The set~$A(x) \in \mathcal{A}$ of active edges for a flow $x$ is called the \emph{support} of the flow.
\end{definition}

\begin{restatable}{lemma}{lempiecewiselinear}
   \label{lem:piecewise-linear}
    For a single-commodity network congestion game, the unique Wardrop equilibrium flow and the cost of the unique Wardrop equilibrium are piecewise linear in $\mu$. In particular, for every $A \in \mathcal{A}$, there is a possibly empty polytope $P_A \subseteq \Delta(\Theta)$ such that $P_A = \{\mu \in \Delta(\Theta) \mid A(x^*(\mu)) = A\}$ and $x^*$ and $C$ are affine on $P_A$. 
\end{restatable}

Here, the polytope $P_A$ is empty if $A \in \mathcal{A}$ is never used in a Wardrop equilibrium.

In the proof, we first establish that the Wardrop equilibrium flow and its cost are linear in~$\mu$ when restricted to a fixed support~$A\in\mathcal{A}$.
This follows from the Karush--Kuhn--Tucker optimality conditions \citep[cf.][Theorems~3.25 and 3.27]{Ruszczyski06} and the shortest path potentials which lead to a characterization of a Wardrop equilibrium as a feasible solution to the following equations 
\begin{subequations}
\label{eq:wardrop-equation-text}
\begin{align}
\pi_v + a_e x_e + \sum_{\theta \in \Theta} \mu_{\theta} b_e^{\theta}  &= \pi_w && \text{ for all } e \in A, \label{eq:wardrop-equation-1-text}\\
\sum_{e \in \delta^+(v)} x_e - \sum_{e \in \delta^-(v)} x_e &= \beta_v && \text{ for all } v \in V, \label{eq:wardrop-equation-2-text}\\
\pi_s &= 0, \label{eq:wardrop-equation-3-text}
\end{align}
\end{subequations}
where the value~$\beta_v$ is the balance of the flow at vertex~$v$, i.e., for an $s$--$t$ flow of value $d$, we have $\beta_s = d$, $\beta_t = -d$, and $\beta_v = 0$ for all $v \notin \{s,t\}$. 
We further have the inequalities
\begin{subequations}
\label{eq:wardrop-inequality-text}
\begin{align}
\pi_v + a_e x_e + \sum_{\theta \in \Theta} \mu_{\theta} b_e^{\theta} &\geq \pi_w && \text{ for all } e \in E \setminus A,\label{eq:wardrop-inequality-1-text}\\
x_e &\geq 0 && \text{ for all } e \in E.\label{eq:wardrop-inequality-2-text}
\end{align}
\end{subequations}
Note that conditions $x_e = 0$ for all $e \in E \setminus A$ are implicitly fulfilled due to the restriction of the characterization to $A$. Afterward, we prove that the linear system \eqref{eq:wardrop-equation-1-text}--\eqref{eq:wardrop-equation-3-text} has full rank and demonstrate that the solution space forms a polytope. For the full proof see the appendix.

We obtain from Lemma~\ref{lem:piecewise-linear}, that the cost of the Wardrop equilibrium $C$ (which can be expressed as $d\pi_t$) is a piecewise linear function on $\Delta(\Theta)$.
As a further corollary from the proof of Lemma~\ref{lem:piecewise-linear}, we obtain that for a given $\mu \in \Delta(\Theta)$, the per-unit cost $\pi_t$ of the Wardrop equilibrium is strictly increasing in the demand.
\citet[Proposition 4.2]{CominettiDS24} showed that the per-unit cost is non-decreasing in the demand, and \citet[Corollary 4]{KlimmW22} gave a similar result for series-parallel networks.
The full proof is deferred to the appendix.

\begin{restatable}{corollary}{cormonotone}
\label{cor:monotone}
Fix $\mu \in \Delta(\Theta)$ and let $\pi_t(d)$ be the per-unit cost of the Wardrop equilibrium as a function of $d$. Then $\pi_t(d)$ is strictly increasing.
\end{restatable}

 \section{Full Information Revelation}
	\label{sec:sepa}
As seen in \Cref{ex:concave}, full information revelation is an optimal signaling scheme whenever the cost function is a concave function of the belief.
In this section, we will show, that the cost function has this property for all single-commodity networks when the underlying graph is series-parallel.
Furthermore, for every graph that is not series-parallel, there are cost functions for the edges such that full information revelation is not an optimal signaling scheme.

Formally, a graph $G = (V,E)$ with two designated vertices $s,t \in V$ is a \emph{series-parallel graph} if it either consists of a single edge $E = \{ \{s,t\}\}$ only, or it is obtained by a parallel or serial composition of two series-parallel graphs. For two series-parallel graphs $G_1 = (V_1, E_1)$ and $G_2 = (V_2,E_2)$ with designated vertices $s_1,t_1 \in V_1$ and $s_2,t_2 \in V_2$, 
	the \emph{parallel composition} is the graph $G = (V,E)$ created from the disjoint union of graphs $G_1$ and $G_2$ by merging the vertices $s_1$ and $s_2$ into a new vertex $s$, and merging the vertices $t_1$ and $t_2$ into a new vertex $t$.
	The \emph{serial composition} of $G_1$ and $G_2$ is the graph created from the disjoint union of graphs $G_1$ and $G_2$ by merging vertices~$t_1$ and $s_2$, and renaming $s_1$ to $s$ and $t_2$ to $t$. In the following, we treat series-parallel graphs as directed graphs by directing every edge in the orientation as it appears in any path from $s$ to $t$. This is well-defined since in a series-parallel graph, there is a global order on the vertices such that every path only visits vertices in increasing order.

To show that full information revelation is optimal, we will show that $C$ is always concave in the belief $\mu$. In Lemma~\ref{lem:piecewise-linear}, we have shown that $C$ is affine on $P_A$ for all $A \in \mathcal{A}$, i.e., there are affine functions $C_A : \Delta(\Theta) \to \mathbb{R}_{\geq 0}$ such that $C_A(\mu) = C(\mu)$ for all $\mu \in P_A$. Intuitively, for $\mu \in P_A$, $C_A(\mu)$ is the unique solution $\pi_t$ to the system of equations \eqref{eq:wardrop-equation}. Let furthermore $x^*_A : \Delta(\Theta) \to \mathbb{R}^E$ be an affine function such that $x^*(\mu) = x^*_A(\mu)$ for all $\mu \in P_A$. Again for $\mu \in P_A$, $x^*_A$ is the unique solution $x$ to the system of equations \eqref{eq:wardrop-equation}.
It is important to note that while $x^*(\mu) = x^*_A(\mu)$ for all $\mu \in P_A$, the vector $x^*_A(\mu)$ will not be a Wardrop equilibrium or even not be a feasible flow at all when $\mu \in\Delta(\Theta) \setminus P_A$. The reason for this is that for $\mu \in \Delta(\Theta) \setminus P_A$ one of the inequalities in \eqref{eq:wardrop-inequality} is violated. If an inequality of type \eqref{eq:wardrop-inequality-1-text} is violated (but none of the other inequalities), the flow is feasible, but not a Wardrop equilibrium since there is an edge outside of the support that would decrease the cost. On the other hand, if an inequality of type \eqref{eq:wardrop-inequality-2-text} is violated, then the flow is not feasible.

In the following, we show that the pointwise minimum $\min_{A \in \mathcal{A}} C_A(\mu)$ of all Wardrop equilibria costs always corresponds to a feasible support.
More specifically, we show that when for some $\mu \in \Delta(\Theta)$, there is a support $A \in \mathcal{A}$ with $\mu \notin P_A$, then there is another support $A' \in \mathcal{A}$ with $C_{A'}(\mu) < C_{A}(\mu)$, or $C_{A'}(\mu) = C_{A}(\mu)$ and $\mu\in P_{A'}$, i.e., $A'$ is a feasible support. 
The proof works by induction on the number of edges where it exploits the structure of series-parallel graphs by distinguishing cases based on whether the last composition was serial or parallel.
For a serial composition, the analysis relies on the cost additivity over the two components, while for a parallel composition, we use the flow monotonicity as stated in \Cref{cor:monotone}.
The full proof is deferred to the appendix.

    \begin{restatable}{lemma}{lemsmallersupport}
    \label{lem:smaller-support}
        For a series-parallel graph, let $A \in \mathcal{A}$ and $\mu \in \Delta(\Theta) \setminus P_A$. Then, there is another support $A' \in \mathcal{A}$ with $C_{A'}(\mu) < C_{A}(\mu)$ or $C_{A'}(\mu) = C_A(\mu)$ and $\mu\in P_{A'}$, i.e., $A'$ is a feasible support.
    \end{restatable}

As a consequence of this lemma, we obtain that the pointwise minimum $\min_{A \in \mathcal{A}} C_A(\mu)$ is always realised by a flow that is a Wardrop equilibrium.
	
\begin{lemma}\label{lem:ward-equ-sepa}
We have $C(\mu) = \min_{A \in \mathcal{A}} C_A(\mu)$ for all $\mu \in \Delta(\Theta)$.
\end{lemma}

\begin{proof}
Let $\mu \in \Delta(\Theta)$ be arbitrary and let $A \in \mathcal{A}$ be such that $C_A(\mu) \leq C_{A'}(\mu)$ for all $A' \in \mathcal{A}$. As shown in Lemma~\ref{lem:smaller-support}, if $x^*_A$ violates one of the inequalities in \eqref{eq:wardrop-inequality}, 
then, there is another support $A'' \in \mathcal{A}$ with $C_{A''}(\mu) < C_{A}(\mu)$ or $C_{A''}(\mu) = C_A(\mu)$ and $A''$ is feasible.
Since $C_{A''}(\mu) < C_{A}(\mu)$ contradicts the choice of $A$, the support $A''$ must be feasible with $C_{A''}(\mu) = C_A(\mu)$.
\end{proof}
	
We obtain the main result of this section.	
	
	\begin{theorem}\label{thrm:sepa-full-opt}
		Full information revelation is an optimal signaling scheme for single-commodity congestion games on series-parallel graphs with affine costs and unknown offsets.
	\end{theorem}
	
	\begin{proof}
A signaling scheme $\Phi$ is a convex decomposition of $\prior$ into distributions $\mu_\sigma \in \Delta(\Theta)$, and $\cost(\Phi)$ is a convex combination of $\cost(\mu_\sigma)$, i.e.,
\begin{align*}
C(\Phi) =   \sum_{\sigma \in \Sigma} \phi_{\sigma} \, C(\mu_{\sigma}).
\end{align*}
Since $C(\mu)$ is concave in $\mu$ the best convex decomposition of the prior occurs when $\mu_{\sigma_j} = \chi_{\theta_j}$ for all $j \in [|\Theta|]$, where $\chi_{\theta_j}$ is the indicator vector for state $\theta_j$.
    \end{proof}
	
	To show that this characterization is indeed tight, we first give an example of a network that is not series-parallel and where full information revelation is not optimal.
		
	\begin{example}\label{example:braess}
		Consider the graph shown in Figures~\ref{fig:braess1} and \ref{fig:braess2}. The cost functions are
		\begin{align*}
			c_e^{\theta_j}(x)=
			\begin{cases}
				x& \text{if }e\in \{e_1,e_4\},\, j=1,2,\\
				1& \text{if }e\in \{e_2,e_3\},\, j=1,2,\\
				0& \text{if }e=e_5,\, j=1,\text{ and }\\
				1& \text{if }e=e_5,\, j=2.
			\end{cases}
		\end{align*} 
		Note that only the cost of edge $e_5$ changes in the different states. 		
		We claim that full information revelation is suboptimal in this example.
		Consider the two different supports
		$A_1=\{e_1,e_2,e_3,e_4\}$ and $A_2=\{e_1,e_2,e_3,e_4,e_5\}$.
		For support $A_1$, we obtain the Wardrop equilibrium $x_{e_1} = x_{e_2} = x_{e_3} = x_{e_4} = 1/2$ for all $\mu \in \Delta(\Theta)$ and for support $A_2$ we obtain the flow $x_{e_1} = x_{e_4} = 1 - \mu_{\theta_2}$, $x_{e_2} = x_{e_3} = \mu_{\theta_2}$, and $x_{e_5} = 1-2\mu_{\theta_2}$, for all  $\mu \in \Delta(\Theta)$.
		This yields the following cost for the supports:
		\begin{align*}
			\cost_{A_1}(\mu)
			&= 2 \cdot \frac{1}{2}\cdot \frac{1}{2} + 2\cdot \frac{1}{2}		= \frac{3}{2},\\
			\cost_{A_2}(\mu)
			&= 2 \cdot (1-\mu_{\theta_2}) \cdot (1-\mu_{\theta_2})
			+ 2\cdot \mu_{\theta_2} + (1- 2\mu_{\theta_2})\cdot \mu_{\theta_2}
			=2-\mu_{\theta_2}.
		\end{align*}
		However, for any $\mu_{\theta_2}<1/2$, the flow defined by support $A_1$ is not a Wardrop equilibrium for the whole graph, since edge $e_5$ violates an inequality of type \eqref{eq:wardrop-inequality-1-text}.
		Furthermore, for any $\mu_{\theta_2}>1/2$, the flow defined by support $A_2$ is not a Wardrop equilibrium for the whole graph since the edge $e_5$ violates an inequality of type \eqref{eq:wardrop-inequality-2-text}.
		Hence, the pointwise minimum is not a concave function, as illustrated in \Cref{fig:braess-cost}. For a prior of $\mu^* = (1/2,1/2)$, full information revelation would yield a cost of $1/2\cdot 2 + 1/2 \cdot 3/2 = 7/4$. In contrast, sending a single signal in all states does not reveal any information about the state, so that the Wardop equilibrium $x^*(\mu^*)$ with cost of $3/2$ emerges. Hence, full information revelation is not optimal.  \hfill $\blacksquare$
	\end{example}
	
\begin{figure}[t!]
\scriptsize
\begin{subfigure}[b]{0.3\textwidth}
\begin{center}
\begin{tikzpicture}[xscale=10/3,yscale=10/3]
\node[state,fill=Green,draw=Green,label=left:{$s$}] (s) at (-0.5,0.4) {}; 
\node[state,fill=Green,draw=Green,label=right:{$t$}] (t) at (0.5,0.4) {}; 
\node[state,fill=Green,draw=Green] (v1) at (0,0.8) {}; 
\node[state,fill=Green,draw=Green] (v2) at (0,0) {}; 
				
\path[->,Green] (s) edge node[below right] {$e_1$} node[above left] {\textcolor{black}{$x$}} (v1);
\path[->,Green] (s) edge node[above right] {$e_3$} node[below left] {\textcolor{black}{$1$}} (v2);
\path[->,Green] (v1) edge node[below left] {$e_2$} node[above right] {\textcolor{black}{$1$}} (t);
\path[->,Green] (v2) edge node[above left] {$e_4$} node[below right] {\textcolor{black}{$x$}} (t);
\path[->,Green] (v1) edge node[right] {$e_5$} node[left] {\textcolor{black}{$0$}}(v2);		
\end{tikzpicture}
\end{center}
\caption{}
\label{fig:braess1}
\end{subfigure}
\begin{subfigure}[b]{0.3\textwidth}
\begin{center}
\begin{tikzpicture}[xscale=10/3,yscale=10/3]
\node[state,fill=Red,draw=Red,label=left:{$s$}] (s) at (-0.5,0.4) {}; 
\node[state,fill=Red,draw=Red,label=right:{$t$}] (t) at (0.5,0.4) {}; 
\node[state,fill=Red,draw=Red] (v1) at (0,0.8) {}; 
\node[state,fill=Red,draw=Red] (v2) at (0,0) {}; 
				
\path[->,Red] (s) edge node[below right] {$e_1$} node[above left] {\textcolor{black}{$x$}} (v1);
\path[->,Red] (s) edge node[above right] {$e_3$} node[below left] {\textcolor{black}{$1$}} (v2);
\path[->,Red] (v1) edge node[below left] {$e_2$} node[above right] {\textcolor{black}{$1$}} (t);
\path[->,Red] (v2) edge node[above left] {$e_4$} node[below right] {\textcolor{black}{$x$}} (t);
\path[->,Red] (v1) edge node[right] {$e_5$} node[left] {\textcolor{black}{$1$}} (v2);		
\end{tikzpicture}
\end{center}
\caption{}
\label{fig:braess2}
\end{subfigure}
\begin{subfigure}[b]{0.39\textwidth}
\begin{center}
\begin{tikzpicture}[shorten > = 0.5pt]
\draw[axis,->] node[below left]{0} (-0.25,0) -- (4.5,0) node[right]{$\mu_{\theta_2}$};
\draw[axis,->] (0,-0.25) -- (0,2.5) node[above]{$C(\mu)$};
				
\draw (4,4pt) -- (4,-4pt) node[below]{1};
\draw (2,4pt) -- (2,-4pt) node[below]{1/2};
\draw (4pt,2) -- (-5pt,2) node[left]{2};
\draw (4pt,3/2) -- (-5pt,3/2) node[left]{3/2};
				
\draw[ultra thick,MidnightBlue] (0,2) -- node[above]{$C_{A_2}(\mu)$} (2,3/2);
\draw[ultra thick,MidnightBlue] (2,3/2) -- (4,1)[dashed];
\draw[ultra thick,CornflowerBlue] (0,3/2) --  (2,3/2)[dashed];
\draw[ultra thick,CornflowerBlue] (2,3/2) -- node[above] {$C_{A_1}(\mu)$} (4,3/2);
\fill[Green] (0,2) circle (3.0pt);
\fill[Green] (0,0) circle (3.0pt);
\draw[thick,Green] (0,2) -- (0,0);
\fill[Red] (4,3/2) circle (3.0pt);
\fill[Red] (4,0) circle (3.0pt);
\draw[thick,Red] (4,3/2) -- (4,0);
\end{tikzpicture}
\end{center}
\caption{}
\label{fig:braess-cost}
\end{subfigure}
\caption{\label{fig:braess}Braess network: (a) cost functions for state $\theta_1$; (b) cost functions for state $\theta_2$; (c) cost of a Wardrop equilibrium as a function of the belief, described by parameter $\mu_{\theta_2} = 1 - \mu_{\theta_1} \in [0,1]$, the costs are dashed where the supports are infeasible for the given $\mu_{\theta_2}$.}
\end{figure}

In general, a Braess network is not embedded in a two-terminal graph network $G$ if and only if $G$ is series-parallel ~\citep{Milchtaich06}. 
In particular, we can define cost functions that essentially reduce any non-series-parallel graph to \Cref{example:braess} which proves the main theorem of this section. See the appendix for the full proof.

\begin{restatable}{theorem}{thrmsepafulloptiff}
   \label{thrm:sepa-full-opt-iff}
    For a single-commodity network $G$ with affine costs and unknown offsets, full information revelation is always an optimal signaling scheme if and only if $G$ is series-parallel. 
\end{restatable}

\section{Computing Optimal Signaling Schemes}
\label{sec:LPs}

	In this section, we consider general multi-commodity network congestion games and affine costs with state-based offsets, and show how optimal signals for these networks can be computed with LP techniques. Towards this end, we first need to investigate the unique Wardrop equilibria for a fixed set of active edges for each commodity. We again use the term \emph{support} for a set of active edges.
 
\subsection{Optimal Signaling via Optimal Support Selection}
	\label{ssec:reduction}

Suppose we are given a set of~$k$ distinct support vectors for each commodity~$i$ denoted by $(\support_{i,1})_{i \in R},\ldots,(\support_{i,k})_{i \in R}$. Consider the set of signaling schemes $\Phi$ with the following properties: $\Phi$ sends $k$ signals (where for simplicity we assume $\sigma \in [k] = \{1,\ldots,k\}$), and each signal $\sigma \in [k]$ results in a Wardrop equilibrium $x_\sigma$ with support $\support_{i,\sigma}$, for every commodity $i \in R$. The main result in this section shows that we can efficiently optimize over this set of signaling schemes.

\begin{restatable}{theorem}{thmbigLP}
    \label{thm:bigLP}
	Given $k$ distinct support vectors $(\support_{i,\sigma})_{i \in R, \sigma \in [k]}$, there exists an LP that computes the best signaling scheme that induces Wardrop equilibria with supports $(\support_{i,\sigma})_{i \in R, \sigma \in [k]}$ in time polynomial in $|\Theta|$, $|E|$, $|R|$, and $k$.
\end{restatable}

\begin{proof}{Proof Sketch.}
We extend the Wardrop flow characterization for a single signal with vertex potentials and flow-conservation constraints to generalize the polytope $P_A$ described by~\eqref{eq:wardrop-equation-text} and \eqref{eq:wardrop-inequality-text} to multi-commodity network congestion games using the balance vector 
\[
        \beta_{v,i} = \begin{cases} \phantom{-}d_i & \text{ if } v = s_i,\\ -d_i & \text{ if } v = t_i, \text{ and} \\ \phantom{-}0 & \text{ otherwise,} \end{cases}
    \]
for each commodity $i \in R$ to define the system of inequalities
    \begin{align}
        \label{eq:wardrop-inequality-multi-text}
        \begin{aligned}
            \pi_{v,i,\sigma} + a_e x_{e,\sigma} + \sum_{\theta \in \Theta} \mu_{\theta,\sigma} b_e^{\theta}  &= \pi_{w,i,\sigma} && \text{ for all } e=(v,w) \in \support_{i,\sigma}, i \in R,\\
            \pi_{v,i,\sigma} + a_e x_{e,\sigma} + \sum_{\theta \in \Theta} \mu_{\theta,\sigma} b_e^{\theta} &\geq \pi_{w,i,\sigma} && \text{ for all } e=(v,w) \in E \setminus \support_{i,\sigma}, i \in R,\\
            \sum_{e \in \delta^+(v)} x_{e,i,\sigma} - \sum_{e \in \delta^-(v)} x_{e,i,\sigma} &= \beta_{v,i} && \text{ for all } v \in V, i \in R,\\
            x_{e,\sigma} &= \sum_{i \in R} x_{e,i,\sigma} && \text{ for all } e \in E,\\ 
            x_{e,i,\sigma} &\geq 0 && \text{ for all } e \in E, i \in R, \\
            \pi_{s_i,i,\sigma} &= 0.
        \end{aligned}
    \end{align}
Again, the conditions $x_{e,i\sigma} = 0$ for all $e \in E \setminus A_{i,\sigma}$ are implicitly fulfilled. These constraints are extended to incorporate conditional beliefs and signaling probabilities, leading to a reformulation where nonlinear constraints are eliminated by substituting $\mu_{\theta,\sigma} = \phi_{\theta,\sigma}/\phi_{\sigma}$ and introducing new variables $y_{e,i,\sigma} = x_{e,i,\sigma} \cdot \phi_{\sigma}$ and $\tau_{v,i,\sigma} = \pi_{v,i,\sigma} \cdot \phi_{\sigma}$. This transformation results in a new polytope that captures Wardrop equilibrium conditions under signaling and enables optimization over signaling schemes via a final linear program (LP) that minimizes total expected cost while maintaining flow conservation and belief consistency. This final LP is given by
\begin{equation}
		\label{eq:schemeSupportNetwork_text}
		\begin{array}{lrcll}
			\text{Min. } & \multicolumn{3}{l}{\D \sum_{\sigma \in [k]} \sum_{i \in R} d_i \cdot \tau_{t_i,i,\sigma} } \\
			\text{s.t. } &\D \tau_{v,i,\sigma} + a_e y_{e,\sigma} + \sum_{\theta \in \Theta} \phi_{\theta,\sigma} b_e^\theta & = & \tau_{w,i,\sigma} & \text{for all } e=(v,w) \in \support_{i,\sigma}, i \in R, \sigma \in [k], \\
			&\D \tau_{v,i,\sigma} + a_e y_{e,\sigma} + \sum_{\theta \in \Theta} \phi_{\theta,\sigma} b_e^\theta & \ge & \tau_{w,i,\sigma} & \text{for all } e=(v,w) \in E \setminus \support_{i,\sigma}, i \in R, \sigma \in [k],\\
	   	    &\D \sum_{e \in \delta^+(v)} y_{e,i,\sigma} - \sum_{e \in \delta^{-}(v)} y_{e,i,\sigma} &=& \beta_{i,v} \cdot \D \sum_{\theta \in \Theta} \phi_{\theta,\sigma} \phantom{~~~}& \text{for all } v \in V, i\in R,\\
			&y_{e,\sigma} &=& \D \sum_{i \in R} y_{e,i,\sigma} & \text{for all } e\in E, \sigma \in [k],\\
			&y_{e,i,\sigma} & \ge & 0  & \text{for all } e \in \support_i, i \in R, \sigma \in [k], \\
		    &\tau_{s_i,i,\sigma} & = & 0 & \text{for all } i \in R, \sigma \in [k], \\
			&\D \sum_{\sigma \in [k]} \phi_{\theta,\sigma} & = & \prior_\theta & \text{for all } \theta\in \Theta,\\
			&\phi_{\theta,\sigma} & \le & \prior_\theta & \text{for all } \theta\in \Theta, \sigma \in [k],\\
			&\phi_{\theta,\sigma} & \ge & 0  & \text{for all } \theta\in \Theta, \sigma \in [k].
		\end{array}
	\end{equation}
Since the number of constraints and variables is polynomial in~$|\Theta|$, $|E|$, $|R|$, and $k$, the LP can be solved efficiently. We refer to the appendix for the full proof. 
\end{proof}

	This reduces optimizing the signaling scheme to an optimal choice of support sets. Suppose for some optimal signaling scheme $\Phi^*$, we know (a superset of) the support combinations $(\support_{i,\sigma})_{i \in R}$ for all commodities in the Wardrop equilibrium resulting from each signal $\sigma \in \Sigma$ issued in $\Phi^*$. Then, we can recover $\Phi^*$ by solving LP~\eqref{eq:schemeSupportNetwork_text}.

    We proceed to inspect the conditions of optimal schemes $\Phi^*$ more thoroughly. Indeed, we can restrict ourselves to at most $k \le |\Theta|$ signals, and each signal $\sigma$ can be assumed to have a distinct support vector $(\support_{i,\sigma})_{i \in R}$.  
 
     \begin{proposition}
		There is an optimal signaling scheme $\Phi^*$ such that
		\begin{enumerate}
			\item at most $|\Theta|$ signals are issued in $\Phi^*$ and
			\item there is no pair of signals $\sigma \neq \sigma'$ that are both issued in $\Phi^*$ and $\support_{i,\sigma} \subseteq \support_{i,\sigma'}$ for each commodity $i \in R$. In particular, every signal $\sigma$ that is issued in $\Phi^*$ has a distinct support vector $(\support_{i,\sigma})_{i \in R}$.
		\end{enumerate}
	\end{proposition}
	
	\begin{proof}
	The first property is a direct consequence of Caratheodory's theorem applied in the context of signaling \citep[cf.][]{Dughmi19}.
	For the second property, consider an optimal signaling scheme $\Phi^*$ resulting from an optimal solution of LP~\eqref{eq:schemeSupportNetwork_text}. Suppose $\Phi^*$ issues two signals $\sigma, \sigma'$ such that $A_{i,\sigma'}\subseteq A_{i,\sigma}$ for all $i\in R$. Then, we can discard signal $\sigma'$ and define a new signaling scheme $\Phi'$ with $\phi'_{\theta,\sigma} = \phi_{\theta,\sigma} + \phi_{\theta,\sigma'}$ and $\phi'_{\theta,\sigma'} = 0$ for every $\theta \in \Theta$. All other probabilities are the same as in $\Phi^*$.
	Similarly, $y'_{e,i,\sigma'} = 0$ and $y'_{e,i,\sigma} = y_{e,i,\sigma} + y_{e,i,\sigma'}$, as well as $\tau'_{v,i,\sigma'} = 0$ and $\tau'_{v,i,\sigma} = \tau_{v,i,\sigma} + \tau_{v,i,\sigma'}$, for every $e \in E$, $v \in V,$ and $i \in R$. This results in a feasible solution of LP~\eqref{eq:schemeSupportNetwork_text} with the same objective function value, i.e., $\Phi'$ is also an optimal signaling scheme.
	
    \end{proof}

	\subsection{Support Enumeration}
        \label{ssec:supportEnum}
    
	The result in Theorem~\ref{thm:bigLP} shows that the main difficulty in finding an optimal signaling scheme is to determine an optimal collection of \emph{supports} used by the Wardrop equilibrium for the conditional beliefs in an optimal signaling scheme. If we have a polynomial-sized superset of the supports used by the optimal signaling scheme, the scheme can be computed by solving LP~\eqref{eq:schemeSupportNetwork_text}.
	
	Each support is a subset of the edges $E$. Even when we consider only a (symmetric) game with a single commodity, up to ${{2^{|E|}} \choose {|\Theta|}}$ many different collections for the supports of (at most) $|\Theta|$ signals can exist. This upper bound can grow by another exponential factor in the number $r$ of commodities.
 
	Let us first concentrate on the case of two states, i.e., $\Theta = \{\theta_1,\theta_2\}$. Instead of considering all subsets of edges, we systematically search through the supports resulting from all possible beliefs $\mu \in \Delta(\Theta)$.
 The cost function for an edge $e \in E$ is given by $c_e(x \mid \mu) = a_e \, x + b_e(\mu)$ with offsets $b_e(\mu) = (1-\mu_{\theta_2}) \, b_e^{\theta_1} + \mu_{\theta_2} b_e^{\theta_2}$.
 	
	\begin{proposition}
		The set of all supports of Wardrop equilibria for all $\mu \in \Delta(\Theta)$ in games with two states can be computed in output-polynomial time.
		\label{prop:polyTwo}
	\end{proposition}
	
	\begin{proof}
        We describe a routine \textsc{ComputeSupport}$(a,b)$ for $0 \le a < b \le 1$. It considers a subinterval $[a,b] \subseteq [0,1]$ and recursively computes the supports of Wardrop equilibria for all distributions $\mu_\alpha = (1-\alpha,\alpha)$ for $\alpha \in [a,b]$ as follows. The routine first draws a value $\alpha \sim U(a,b)$ uniformly at random from (the interior of) the interval. Then, we compute the Wardrop equilibrium $x$ for belief $\mu_{\alpha}$ in polynomial time. This determines the cost $c_e(x_e \mid \mu_\alpha)$ for every edge $e \in E$. Using a shortest-path computation starting from $s_i$, we can identify the active edges for commodity $i\ \in R$, and, hence, the corresponding support $(\support_i)_{i \in R}$ of $x$ in polynomial time. For this support, we consider the LP composed of constraints \eqref{eq:wardrop-inequality-multi-text} as well as the standard distributional constraints 
        \begin{equation}
           \label{eq:2dist}
            \mu_{\theta_1,\sigma} = 1-\mu_{\theta_2,\sigma} \in [a,b].
        \end{equation}
        We solve this LP twice: once with the objective of maximizing $\mu_{\theta_1}$ and once with the objective of minimizing $\mu_{\theta_1}$. By non-degeneracy (with probability 1), we obtain values $\alpha_1 > \alpha_2$, such that supports $\support_i$ are used by the Wardrop equilibrium for all distributions $\alpha \in [\alpha_1, \alpha_2]$. 
        
        In this way, we obtain three subintervals: $[a,\max\{a,\alpha_1\}]$, $[\max\{a,\alpha_1\}, \min\{\alpha_2,b\}]$ and $[\min\{\alpha_2,b\},b]$. In the middle interval, the supports $\support_i$ are used. We compute the supports in the other two intervals if they are non-degenerate, i.e., we call \textsc{ComputeSupport}$(a,\alpha_1)$ if $\alpha_1 > a$, and \textsc{ComputeSupport}$(\alpha_2,b)$ if $\alpha_2 < b$. Since all distributions with the same support form a convex set, there is at most one consecutive subinterval of $[0,1]$ corresponding to each support. As such, each call of \textsc{ComputeSupport} will generate at least one additional support.
    \end{proof}
	
	The proposition shows that for every game where the Wardrop equilibria for beliefs in $\Delta(\Theta)$ use at most a polynomial number of different supports, we can compute these supports and, thus, an optimal signaling scheme in polynomial time. 
 
    However, this property is not always fulfilled. In particular, even in a symmetric game with a single commodity $R = \{1\}$ and two states that differ only in the offset of a single edge, an exponential number of supports may arise. 
    The proof uses a class of games on a slight variation of the family of nested Braess graphs as defined in \cite{KlimmW22}. It is deferred to the appendix.

\begin{restatable}{theorem}{thmexpoentialsupports}
\label{thm:expoential-supports}
For every number $n \in \NN$, there is a symmetric network congestion game with $|\Theta| = 2$ states, $O(n)$ vertices, $O(n)$ edges, and $O(n)$ source-target paths, in which $2^{n+1}-1$ 
        different supports arise in the Wardrop equilibria for all $\mu \in \Delta(\Theta)$.
\end{restatable}

\subsection{Support Enumeration for Parallel Edges}

In the following, we will work towards showing that for network congestion games on $m$ \emph{parallel} edges only a polynomial number of supports arise.
A consequence of Theorem~\ref{thrm:sepa-full-opt-iff} is that for a \emph{single-commodity} instance on a parallel edges network, full information revelation is always optimal. In contrast to this result, we show that for two commodities on a parallel edges network full information revelation need not be optimal.

\begin{example}
\label{ex:counter-multi-commodity}
There are two vertices $V=\{s,t\}$, two parallel edges $E = \{e_1,e_2\}$, and two commodities $R=\{1,2\}$. Each commodity has volume $d_1 = d_2 = 1/2$. Commodity~1 is restricted to $\calS_1 = \{e_1\}$. Commodity~2 can route on both edges, i.e., $\calS_2 = \{e_1,e_2\}$. There are two states $\theta_1$ and $\theta_2$. The cost functions are $c_{e_1}^{\theta_1}(x) = x$ and $c_{e_2}^{\theta_1}(x) = 1$, as well as $c_{e_1}^{\theta_2}(x) = x+1$ and $c_{e_2}^{\theta_2}(x) = 1/2$; see Figures~\ref{fig:counterexample1a} and \ref{fig:counterexample1b}.
The prior is $\prior_{\theta_1} = \prior_{\theta_2} = 1/2$.

Consider a belief $\mu = (\mu_{\theta_1}, \mu_{\theta_2})$ and denote by $x_{e,i}(\mu)$ the flow of commodity~$i$ on edge~$e$ as a function of $\mu$.
Simple computations show that
\begin{align*}
x_{e_1,1}(\mu) &= \frac{1}{2}, &
x_{e_1,2}(\mu) &= \max \biggl\{ \frac{1}{2} - \frac{3}{2} \mu_{\theta_2},0 \biggr\},  &
x_{e_2,2}(\mu) &= \min \biggl\{\frac{3}{2} \mu_{\theta_2}, \frac{1}{2}\biggr\}.
\end{align*}
The corresponding cost of the Wardrop equilibrium is
\begin{align*}
 C(\mu) = 
 \begin{cases}
- \frac{1}{2}\mu_{\theta_2}^2  + 1 & \text{ if } \mu_{\theta_2} \leq \frac{1}{3}, \\
\frac{1}{4}\mu_{\theta_2} + \frac{3}{4} & \text { if } \mu_{\theta_2} > \frac{1}{3},
 \end{cases}
\end{align*}
and is depicted in Figure~\ref{fig:counterexample1c}.
For the full information revelation signaling scheme $\Phi_{\text{FI}}$ we obtain $C(\Phi_{\text{FI}}) = 1$.
For the signaling scheme $\Phi_{\text{NO}}$ where the same signal is always issued, we obtain $C(\Phi_{\text{NO}}) = C(1/2) = 7/8 < 1$. Hence, full information revelation need not be optimal for networks of parallel edges with multiple commodities. \hfill $\blacksquare$
\end{example}

 \begin{figure}
 \scriptsize
 \begin{center}
\begin{subfigure}[b]{0.25\textwidth}
\begin{center}
\begin{tikzpicture}[scale = 0.7]
\useasboundingbox (0,-2.5) rectangle (4,2); 
\node[state,draw=Green,fill=Green,label=left:{$s$}] (s) at (0,0) [circle] {}; 
\node[state,draw=Green,fill=Green,label=right:{$t$}] (t) at (4,0) [circle] {}; 
				
\path[->,Green] (s) edge[bend left=60]  node[above,yshift=0mm] {\textcolor{black}{$x$}} node[below] {\textcolor{black}{$r \in \{1,2\}$}} (t); 
\path[->,Green] (s) edge[bend right=60]  node[below,yshift=0mm] {\textcolor{black}{$1$}} node[above] {\textcolor{black}{$r = 2$}} (t); 
\end{tikzpicture}
\end{center}
\caption{\label{fig:counterexample1a}}
\end{subfigure}
\begin{subfigure}[b]{0.25\textwidth}
\begin{center}
\begin{tikzpicture}[scale=0.7]
\useasboundingbox (0,-2.5) rectangle (4,2); 
\node[state,draw=Red,fill=Red,label=left:{$s$}] (s) at (0,0) [circle] {}; 
\node[state,draw=Red,fill=Red,label=right:{$t$}] (t) at (4,0) [circle] {}; 
				
\path[->,Red] (s) edge[bend left=60]  node[above,yshift=0mm] {\textcolor{black}{$x+1$}} node[below] {\textcolor{black}{$r \in \{1,2\}$}} (t); 
\path[->,Red] (s) edge[bend right=60]  node[below,yshift=0mm] {\textcolor{black}{$1/2$}} node[above] {\textcolor{black}{$r = 2$}} (t); 
\end{tikzpicture}
\end{center}
\caption{\label{fig:counterexample1b}}
\end{subfigure}
\begin{subfigure}[b]{0.44\textwidth}
\begin{center}
\begin{tikzpicture}[xscale=4,yscale=2]
\draw[axis,->] node[below left]{0} (-1pt,0) -- (1.125,0) node[right]{$\mu_{\theta_2}$};
\draw[axis,->] (0,-2pt) -- (0,1.25) node[above]{$C(\mu)$};
				
\draw (1,2pt) -- (1,-2pt) node[below]{1};
\draw (1/3,2pt) -- (1/3,-2pt) node[below]{1/3};
\draw (1pt,1) -- (-1pt,1) node[left]{1};
\draw (1pt,7/8) -- (-1pt,7/8) node[left]{7/8};
				
\draw[ultra thick,MidnightBlue,smooth,domain=0:1/3,variable=\x] plot({\x},{-1/2*\x + 1});
\draw[ultra thick,CornflowerBlue,smooth,domain=1/3:1,variable=\x] plot({\x},{1/4*\x+3/4});
\node[circle,draw=none,fill=Green,inner sep=0pt, minimum size=6pt] (cg1) at (0,1) {};
\node[circle,draw=none,fill=Green,inner sep=0pt, minimum size=6pt] (cg2) at (0,0) {};
\draw[thick,Green] (cg1) -- (cg2);
\node[circle,draw=none,fill=Red,inner sep=0pt,minimum size=6pt] (cr1) at (1,1) {};
\node[circle,draw=none,fill=Red,inner sep=0pt,minimum size=6pt] (cr2) at (1,0) {};
\draw[thick,Red] (cr1) -- (cr2);
\end{tikzpicture}
\end{center}
\caption{\label{fig:counterexample1c}}
\end{subfigure}
\end{center}
\caption{Example of an instance on parallel edges with two commodities where full information revelation is suboptimal.
\label{fig:counterexample1}}
 \end{figure}

Having thus established that the design of the optimal signaling scheme is non-trivial for parallel edges networks with more than one commodity, we apply our support enumeration approach to compute an optimal signaling scheme in this setting.
More specifically, we obtain a polynomial-time algorithm when the number of commodities $|R|$ is constant.
In the following, for $\mu \in \Delta(\Theta)$ and $e \in E$, let $b_e(\mu) = \sum_{\theta \in \Theta} \mu_\theta b_e^{\theta}$.
Given $\mu \in \Delta(\Theta)$, we call the permutation that orders the edges non-decreasingly in $b_e(\mu)$ the \emph{offset ordering}.

We first show that as long as the offset ordering of the edges does not change and the number of commodities is constant, there is at most a polynomial number of supports.

\begin{lemma}
\label{lem:constantPopulations}
Consider a network congestion game with $m$ parallel edges and $r=|R|$ commodities. Let $P \subseteq \Delta(\Theta)$ be such that $b_{e_1}(\mu) \leq b_{e_2}(\mu) \leq \dots \leq b_{e_m}(\mu)$ for all $\mu \in P$. Then, there are $O(m^{(2^r-1)} \cdot 2^{r(2^{r-1})})$ different combinations of supports of Wardrop equilibria for all $\mu \in P$.
\end{lemma}

\begin{proof}
First, we classify each edge $e \in E$ based on the subset of commodities $\{ i \in R \mid e \in \mathcal{S}_i\}$ that have access to $e$. Specifically, for $C \subseteq R$ with $C \neq \emptyset$, let $E_C = \{e \in E \mid \{i \in R : e \in \mathcal{S}_i\} = C\}$ be the class of edges available to the subset $C$ of commodities.
This yields $2^r-1$ classes of edges.
For a given class $E_C = \{e_{k_1},e_{k_2},\ldots\}$, there is $k \in \{1,\dots,m\}$ such that $x_{e_{k_j}} > 0$ if and only if $k_j \leq k$
since the edges are sorted by expected offset for all $\mu \in P$.
Hence, within each class, there are at most $m$ possible supports. We conclude that there are at most $m^{(2^r-1)}$ different combinations of supports over all classes. 
		
Given such a combination of supports for all classes of edges, we still have to specify which subset of edges is used by commodity $i$, for every $i=1,\ldots,r$.
Since we fix the support in a class $E_C$, all edges of $E_C$ that are in the support must have the same cost in the Wardrop equilibrium. Hence, for each commodity that uses an edge from class $E_C$ in its individual support, we can assume that it uses all edges of the fixed support of $E_C$ in its individual support.
Therefore, for each commodity $i$, we only have to specify, which \emph{classes} of edges are used by this commodity. There are at most $2^{r-1}$ classes available to commodity~$i$, i.e., a total of $2^{(2^{r-1})}$ possibilities for commodity~$i$ and $(2^{(2^{r-1})})^r = 2^{r(2^{r-1})}$ possibilities in total. Overall, this yields $O(m^{(2^r-1)} \cdot 2^{r(2^{r-1})})$ different supports, as claimed.		
\end{proof}

Using bounds for the number of cells of hyperplane arrangements allows us to bound the number of different offset orderings that can appear in an instance on $m$ parallel edges with $k$ states. We then obtain the following result.
	
\begin{lemma}
\label{lem:par-edges-poly}
For a network congestion game with $m$ parallel edges, $k=|\Theta|$ states, and $r=|R|$ commodities, there are at most $O(m^{2k+(2^r)-1} \cdot 2^{r(2^{r-1})})$ different supports of Wardrop equilibria for all $\mu \in \Delta(\Theta)$.
\end{lemma}
	
\begin{proof}
In light of Lemma~\ref{lem:constantPopulations}, it suffices to bound the number of different offset orderings by $O(m^{2k})$. 
To this end, consider the simplex of $\Delta(\Theta)$, i.e., the space of all distributions over $\Theta$ as a subset of $\R^k$.
Let $(e_1,e_1'),\dots,(e_l,e_l')$ with $l = \frac{m(m-1)}{2}$ be an arbitrary but fixed order of all pairs of edges where we have $e_i \neq e_i'$ for all $i \in\{1,\dots,l\}$ and $\{e_i,e_i'\} \neq \{e_j,e_j'\}$ for all $i,j \in \{1,\dots,l\}$ with $i \neq j$.
We consider the hyperplanes generated by the equation $b_{e_i}(\mu) - b_{e_i'}(\mu) = 0$ for all pairs of edges $i \in\{1,\dots,l\}$.
This yields a hyperplane arrangement of $l$ hyperplanes in $\R^k$.
For such a hyperplane arrangement and a belief $\mu \in \Delta(\Theta)$, let the sign vector $\chi(\mu) = (\chi_{1}(\mu),\dots,\chi_l(\mu)) \in \{-,0,+\}^l$ be defined as
\begin{align*}
 \chi_i(\mu) = \begin{cases}
  +   & \text{ if } b_{e_i}(\mu) - b_{e_i'}(\mu) > 0, \\
  0   & \text{ if } b_{e_i}(\mu) - b_{e_i'}(\mu) = 0, \text{ and} \\
   -   & \text{ if } b_{e_i}(\mu) - b_{e_i'}(\mu) < 0.
 \end{cases}
\end{align*}
Let $Z$ be the set of sign vectors $z$ without a zero entry with the additional property that there is a $\mu \in \Delta(\Theta)$ such that $\chi(\mu) = z$.
Note that if there exists some $i\in \{1,\dots,l\}$ such that $b_{e_i}(\mu) - b_{e_i'}(\mu) = 0$ for all $\mu \in \Delta(\Theta)$, then $\chi_i(\mu) = 0$ for all $\mu \in \Delta(\Theta)$ and $Z = \emptyset$. In that case, we set $\chi_i(\mu) = +$ instead.
A sign vector $z \in Z$ corresponds to a \emph{cell} of the hyperplane arrangement.
By Buck's formula \citep{Buck43}, we have $|Z| \leq \sum_{j=0}^k \binom{l}{j} \in O(l^k) = O(m^{2k})$.
    \end{proof}
	
\begin{theorem}
\label{thm:parallel-poly-time}
	For a network congestion game with $m$ parallel edges, $k=|\Theta|$ states, and $r=|R|$ commodities, where $k$ and $r$ are constant, an optimal signaling scheme can be computed in polynomial time.
\end{theorem}
\begin{proof}
As shown in the proof of \Cref{lem:par-edges-poly}, each offset ordering corresponds to a cell in a hyperplane arrangement with $|Z|\in O(m^{2k})$ cells. These cells can be enumerated with the reverse search algorithm \citep{AvisF96} in time linear in $l$, $k$, $|Z|$, and the time needed to solve a linear program with $k$ variables and $l-1$ inequalities, where $l=\frac{m(m-1)}{2}$.
In \Cref{lem:constantPopulations} we showed that for each offset ordering, there is at most $O(m^{(2^r-1)} \cdot 2^{r(2^{r-1})})$ different combinations of supports of Wardrop equilibria for all $\mu \in P$.
These supports can be enumerated explicitly.
Thus, we may use the LP formulation from \Cref{thm:bigLP} to obtain the claimed result.
\end{proof}

 \section{Computational Studies}
	\label{sec:compStudy}

The goal of the computational studies conducted in this section is to investigate i) whether instances of our model on realistic networks generate a small number of different supports in the Wardrop equilibrium, and ii) by how much public signaling can improve the total cost in these networks.

    \begin{table}[t]
    \footnotesize
    \begin{center}
     \caption{Networks considered in our computational studies.}
    \label{tab:cs_networks}
    \begin{tabular}{lccccc}\toprule
	Network name & Abbreviation & $|V|$ & $|E|$    & $|Z|$  & Demand \\\midrule
	Sioux Falls    & SF & \phantom{2}24 & \phantom{2}76 & 24  & 360,600  \\
	Eastern Massachusetts* & EM  & \phantom{2}74  & 258  & 74 & \phantom{3}65,576 \\
	Berlin-Friedrichshain & BF & 224 & 523  & 23  & \phantom{3}11,205 \\
	Berlin-Prenzlauer Berg-Center & BP  & 352  & 749 & 38 & \phantom{3}16,660 \\
	Berlin-Tiergarten & BT & 361 & 766  & 26  & \phantom{3}10,755  \\
	Berlin-Mitte-Center & BM & 398 & 871 & 36 & \phantom{3}11,482 \\\bottomrule
		*Highway-extract version   &  & & &  &  \\
    \end{tabular}
    \end{center}
    \end{table}

   We consider real-world networks for a single commodity and two possible states of nature $\Theta = \lbrace \theta_1, \theta_2 \rbrace$. Table~\ref{tab:cs_networks} shows the six different networks we examine.
   The network data is obtained from the GitHub repository of the \cite{CSData22}. The data set includes a model for each network, i.e., it specifies vertices $V$ and edges $E$ corresponding to crossings and roads in the real world, respectively.
   It also defines a partition of the vertices into \emph{zones} $Z$.
   The sizes of the networks range from rather small ones (SF, EM) to fairly large ones (BF, BP, BT, BM). The first two are frequently considered in the traffic assignment literature; the latter were used, e.g., by \citet{Jahn05}. 

    In addition, the data set provides information on traffic-related properties for each edge $e \in E$, such as the capacity $u_e$ and the free-flow travel time $t_e$ (i.e., the time needed to traverse the edge in the absence of congestion), and representative demands between pairs of zones. Note that for SF and EM each vertex corresponds to one zone. Originally, the data set is designed for computational studies of traffic assignment with multiple commodities and edge cost functions $c^\text{BPR}_e(x)$ as defined in the congestion model of the Bureau 
    \cite{BPR64},
	\begin{equation}
        \label{eq:cs_bpr}
    	c^\text{BPR}_e(x) = t_e  \left( 1 + \eta  \left( \frac{x}{u_e} \right)^\beta \right) \, ,
	\end{equation}
    where $\beta = 4$, and $\eta$ is a dimensionless parameter ($\eta = 0.15$ for SF and EM, $\eta = 1$ else).
    To obtain affine functions, we set $\beta=1$ and obtain the cost functions $c_e (x) = a_e  x + b_e$ with coefficients 
	\begin{equation}
	    \label{eq:cs_coeffs}
	    a_e = \eta \cdot \frac{t_e }{u_e} \qquad \text{and} \qquad b_e = t_e \, .
	\end{equation}
      In this way, we obtain slopes $a_e$ of orders of magnitude from $10^{-7}$ to $10^{-5}$ for SF and EM and from $10^{-4}$ to $10^{-2}$ for the remaining ones. The offsets $b_e$ are integers ranging from 0 to 60, depending on the instance.      
      The offsets in the original data are deterministic. We introduce a parameter $\tau \in [0,1]$ to control the fraction of edges with stochastic offsets.
      For a $\tau$-fraction of the edges, we set $b_e$ according to \eqref{eq:cs_coeffs} in either $\theta_1$ or $\theta_2$ uniformly at random.
      For the other state, $b_e$ is drawn uniformly at random from the interval $[0,1.25 \cdot \arg \max_{e \in E} \, t_e]$ to allow for deviations to both lower and higher values.
      For the remaining fraction of $1-\tau$ of edges, we keep $b^{\theta_2}_e = b^{\theta_1}_e$.
      For our single-commodity scenario, we set the demand such that it equals the total demand that is routed through the network for the multi-commodity scenario in the original data, as specified in Table~\ref{tab:cs_networks}.   

      In the following, we show results for $\tau = 0.75$. We perform 40 computations for each network with varying $(s,t)$-pairs. For each computation, the $(s,t)$-pair is drawn uniformly at random from the set of zones such that $s \not = t$ and no pair is chosen more than once. Thus, each computation is given one network and one $(s,t)$-pair. We call such a tuple an \emph{instance}.

     To determine the set of all supports of Wardrop equilibria for all $\mu_{\theta_2} \in [0,1]$, we apply the recursive approach described in the proof of Proposition~\ref{prop:polyTwo}, i.e., we recursively compute the support of the emerging Wardrop equilibrium at a mean value for $\mu_{\theta_2}$ (initially $\mu_{\theta_2} = 1/2$), and then solve the LP composed of the constraints in \eqref{eq:wardrop-inequality-multi-text} as well as \eqref{eq:2dist} twice: once to maximize $\mu_{\theta_2}$ and once to minimize $\mu_{\theta_2}$.
     Since for two states, two signals suffice, the optimal signaling scheme can be computed by enumerating over all pairs of breakpoints, i.e., the beliefs where the support of the Wardrop equilibrium changes.     
     The LPs are solved using the built-in solver of the SciPy package (v1.6.2) by \cite{2020SciPy-NMeth}. The flow assignments are computed by an implementation of the conjugate Frank-Wolfe algorithm \citep{FrankWolfe56, Daneva03} in Python (v3.8.10) based on the code of \cite{CSBasisCode22}. The computing platform is an Intel Core i5 based computer running at 3.47 GHz with 8 GB RAM.

    \begin{table}[t]
    \footnotesize
    \begin{center}
   \caption{Results for the set of all supports $\mathcal{A}_i$ and properties of $\mathrm{\cost(\mu)}$ for 40 instances for each network.}
    \label{tab:cs_supports}
    \begin{tabular}{l@{~~~~~~}ccccc}\toprule
	\begin{tabular}{@{}c@{}} Network \end{tabular} & \begin{tabular}{@{}c@{}}Average\\  of $|\mathcal{A}_i|$ \end{tabular} &  \begin{tabular}{@{}c@{}}Standard dev.\  \\ of $|\mathcal{A}_i|$ \end{tabular} & \begin{tabular}{@{}c@{}} Maximum\\  of $|\mathcal{A}_i|$ \end{tabular}  & \begin{tabular}{@{}c@{}} $\mathrm{\cost(\mu)}$ concave  \\ but not linear $[\%]$ \end{tabular}  & \begin{tabular}{@{}c@{}} $\mathrm{\cost(\mu)}$ \\ linear $[\%]$ \end{tabular}\\\midrule
	SF &  3.75 &	2.24 &	11 &	70 & 10 \\ 
	EM &  6.22 &	3.93 &	15 &	65 & \phantom{1}8  \\ 
	BF &  4.85 &	2.76 &	12 &	75 & \phantom{1}3  \\ 
	BP &  4.03 &	2.34 &	11 &	58 & 13   \\ 
	BT &  6.72 &	3.59 &	15 &	53 & \phantom{1}3 \\ 
	BM &  4.35 &	2.49 &	\phantom{1}9 &	55 & 13  \\ 
    \bottomrule 
    \end{tabular}
    \end{center}
    \end{table}

     For each network with instances $i=1, \ldots,40$ let $\mathcal{A}_i$ be the set of all (distinct) supports of $\cost(\mu)$. Table~\ref{tab:cs_supports} shows averaged results on the properties of $\mathcal{A}_i$. We point out that both the average and the maximum number of used supports turn out to be very small compared to the number of edges in each network, even though $\tau$ is rather high. It turns out that these quantities tend to increase as the parameter $\tau$ is raised from 0 to 1.  Moreover, the average standard deviation is small, as well. Therefore, these findings imply that computing the optimal signaling scheme in realistic network instances can be done efficiently by our approach.
     For complementary illustrations, we show the five supports that appear in an instance of SF for the Wardrop equilibrium as well as the optimal signaling scheme in Figure~\ref{fig:sf-example} in Appendix~\ref{app:illu_sf}.

    \begin{table}[t]
    \footnotesize
    \begin{center}
       \caption{Performance of full information revelation (FI), no-signaling (NO), and the optimal signaling scheme~(OPT) averaged over 40 instances for each network. }
    \label{tab:cs_costs}
    \begin{tabular}{l@{~~~~~~}cccc}\toprule
	\begin{tabular}{@{}c@{}} Network \end{tabular} & \begin{tabular}{@{}c@{}} FI is \\ optimal $[\%]$  \end{tabular}  & \begin{tabular}{@{}c@{}} cost(NO) \\\hline cost(OPT)  \end{tabular} & \begin{tabular}{@{}c@{}} cost(OPT) \\\hline cost(Pointwise SO)  \end{tabular} & \begin{tabular}{@{}c@{}} cost(NO) \\\hline cost(Pointwise SO)  \end{tabular}\\\midrule
	SF &  100 & 1.0050 &	1.0176 &	1.0227 \\
    EM &  100 &	1.0819 &	1.0229 &	1.1066 \\
    BF &  100 &	1.0141 &	1.0190 &	1.0334 \\
    BP &  \phantom{1}98 &	1.0028 &	1.0100 &	1.0129 \\
    BT &  \phantom{1}96 &	1.0165 &	1.0186 &	1.0354 \\
    BM &  100 &	1.0083 &	1.0170 &	1.0255 \\
    \bottomrule
    \end{tabular}
    \end{center}
    \end{table}

    For the second part of our study, we analyze the performance of full information revelation, no-signaling (i.e., revealing no information by sending the same signal for all states), and the optimal signaling scheme, as shown in Table~\ref{tab:cs_costs}. The results are rounded to four decimal places due to numerical precision.  We assume the prior $\prior_{\theta_1} = \prior_{\theta_2} = 0.5$ as a reference.
    The cost of no-signaling is equal to the cost of the Wardrop equilibrium for the prior.
    The cost of full information revelation is equal to the cost of the Wardrop equilibria for states $\theta_1$ and $\theta_2$ averaged by their respective prior probabilities.
    We also compare the costs of the signaling schemes to the following lower bound on the cost of any signaling scheme. The \emph{pointwise} social optimum (pointwise SO) is defined by the costs of the system optimal flows for states $\theta_1$ and $\theta_2$ averaged by their respective prior probabilities.
    
    One can see that in most cases full information revelation is optimal.
    Moreover, even when it is not optimal, the additional costs are not notable within the numerical precision applied.
    Since there is no network where $\cost(\mu)$ is convex but non-linear, no-signaling is only optimal in the rare cases that the Wardrop equilibrium is linear (see Table~\ref{tab:cs_supports}). This appears to be the case mainly when the source and target vertices are very close to each other.

    \Cref{tab:cs_supports} shows that a significant proportion of the networks induces non-concave cost functions which empirically confirms our characterization from \Cref{sec:sepa} as the underlying networks are non-series-parallel.
    The (approximate) optimality of full information revelation shown in \Cref{tab:cs_costs} suggest, however, that its optimality persists even in cases where the cost of the Wardrop equilibrium is not concave.
    Indeed, this observation is confirmed by our experiments for more general cost functions in the next section. Moreover, the results reveal that optimal signaling results in slight but consistent improvements over no-signaling. However, even with optimal information design there remains a notable gap to the average cost of a pointwise social optimal flow. The latter is related to the fact that due to the optimality of full information revelation, optimal signaling induces a Wardrop equilibrium pointwise for each state. Thus, comparing optimal signaling and the pointwise social optimal flow, we measure a variant of an ``expected price of anarchy''~\citep{HoeferS10}, which is in similar orders of magnitude as the (deterministic) price of anarchy in these networks~\citep{Jahn05}.

\section{Discussion}
\label{sec:discussion}
In this paper, we studied how the provision of information about the realization of the travel times in networks may help decrease the total travel times in the emerging Wardrop equilibria. 
Specifically, we showed that for a series-parallel network with a single commodity, it is always optimal to provide the full information about the realized states. 
We assume that cost functions are affine with state-based offsets, and it is natural to ask whether this result extends to further classes of cost functions. 
With the following examples, we illustrate that this is not the case.

\begin{figure}
\scriptsize
\begin{subfigure}[b]{0.3\textwidth}
\begin{center}
\begin{tikzpicture}[scale = 0.7]
\useasboundingbox (0,-2.5) rectangle (4,1);
\node[state,draw=Green,fill=Green,label=left:{$s$}] (s) at (0,0) [circle] {}; 
\node[state,draw=Green,fill=Green,label=right:{$t$}] (t) at (4,0) [circle] {}; 
				
\path[->,Green] (s) edge[bend left=60]  node[above,yshift=0mm] {\textcolor{black}{$1$}} (t); 
\path[->,Green] (s) edge[bend right=60]  node[below,yshift=0mm] {\textcolor{black}{$x$}} (t); 
\end{tikzpicture}
\end{center}
\caption{\label{fig:counterexample2a}}
\end{subfigure}
\begin{subfigure}[b]{0.3\textwidth}
\begin{center}
\begin{tikzpicture}[scale=0.7]
\useasboundingbox (0,-2.5) rectangle (4,1);
\node[state,draw=Red,fill=Red,label=left:{$s$}] (s) at (0,0) [circle] {}; 
\node[state,draw=Red,fill=Red,label=right:{$t$}] (t) at (4,0) [circle] {}; 
				
\path[->,Red] (s) edge[bend left=60]  node[above,yshift=0mm] {\textcolor{black}{$x$}} (t); 
\path[->,Red] (s) edge[bend right=60]  node[below,yshift=0mm] {\textcolor{black}{$2$}} (t); 
\end{tikzpicture}
\end{center}
\caption{\label{fig:counterexample2b}}
\end{subfigure}
\begin{subfigure}[b]{0.39\textwidth}
\begin{center}
\begin{tikzpicture}[xscale=4,yscale=2]
\draw[axis,->] node[below left]{0} (-1pt,0) -- (1.125,0) node[right]{$\mu_{\theta_2}$};
\draw[axis,->] (0,-2pt) -- (0,1.25) node[above]{$C(\mu)$};
				
\draw (1,2pt) -- (1,-2pt) node[below]{1};
\draw (0.50,2pt) -- (0.50,-2pt) node[below]{1/2};
\draw (1pt,1) -- (-1pt,1) node[left]{1};
				
\draw[ultra thick,MidnightBlue,smooth,domain=0:0.5,variable=\x] plot({\x},{2*\x*\x - \x + 1});
\draw[ultra thick,CornflowerBlue,smooth,domain=0.5:1,variable=\x] plot({\x},{1});
\node[circle,draw=none,fill=Green,inner sep=0pt, minimum size=6pt] (cg1) at (0,1) {};
\node[circle,draw=none,fill=Green,inner sep=0pt, minimum size=6pt] (cg2) at (0,0) {};
\draw[thick,Green] (cg1) -- (cg2);
\node[circle,draw=none,fill=Red,inner sep=0pt,minimum size=6pt] (cr1) at (1,1) {};
\node[circle,draw=none,fill=Red,inner sep=0pt,minimum size=6pt] (cr2) at (1,0) {};
\draw[thick,Red] (cr1) -- (cr2);
\node[circle,draw=none,fill=Dandelion,inner sep=0pt,minimum size=3pt] (cd1) at (0.29,0.88 ) {};
\node[circle,draw=none,fill=Dandelion,inner sep=0pt,minimum size=3pt] (cd2) at (1,1) {};
\draw[thick,Dandelion] (cd1) -- (cd2);
\end{tikzpicture}
\end{center}
\caption{\label{fig:counterexample2c}}
\end{subfigure}
\caption{A parallel edges network where full information revelation is suboptimal. The orange line in (c) indicates the optimal signaling scheme.}
\end{figure}
\begin{figure}
\scriptsize
\begin{subfigure}[b]{0.3\textwidth}
\begin{center}
\begin{tikzpicture}[scale = 0.7]
\useasboundingbox (0,-2.5) rectangle (4,1);
\node[state,draw=Green,fill=Green,label=left:{$s$}] (s) at (0,0) [circle] {}; 
\node[state,draw=Green,fill=Green,label=right:{$t$}] (t) at (4,0) [circle] {}; 
				
\path[->,Green] (s) edge[bend left=60]  node[above,yshift=0mm] {\textcolor{black}{$x^2$}} (t); 
\path[->,Green] (s) edge[bend right=60]  node[below,yshift=0mm] {\textcolor{black}{$x^2+1$}} (t); 
\end{tikzpicture}
\end{center}
\caption{\label{fig:counterexample3a}}
\end{subfigure}
\begin{subfigure}[b]{0.3\textwidth}
\begin{center}
\begin{tikzpicture}[scale=0.7]
\useasboundingbox (0,-2.5) rectangle (4,1);
\node[state,draw=Red,fill=Red,label=left:{$s$}] (s) at (0,0) [circle] {}; 
\node[state,draw=Red,fill=Red,label=right:{$t$}] (t) at (4,0) [circle] {}; 
				
\path[->,Red] (s) edge[bend left=60]  node[above,yshift=0mm] {\textcolor{black}{$x^2+1$}} (t); 
\path[->,Red] (s) edge[bend right=60]  node[below,yshift=0mm] {\textcolor{black}{$x^2$}} (t); 
\end{tikzpicture}
\end{center}
\caption{\label{fig:counterexample3b}}
\end{subfigure}
\begin{subfigure}[b]{0.39\textwidth}
\begin{center}
\begin{tikzpicture}[xscale=4,yscale=2]
\draw[axis,->] node[below left]{0} (-1pt,0) -- (1.125,0) node[right]{$\mu_{\theta_2}$};
\draw[axis,->] (0,-2pt) -- (0,1.25) node[above]{$C(\mu)$};
				
\draw (1,2pt) -- (1,-2pt) node[below]{1};
\draw (0.5,2pt) -- (0.5,-2pt) node[below]{1/2};
\draw (1pt,1) -- (-1pt,1) node[left]{1};
				
\draw[ultra thick,MidnightBlue,smooth,domain=0:1,variable=\x] plot({\x},{\x*\x - \x + 1});
\node[circle,draw=none,fill=Green,inner sep=0pt, minimum size=6pt] (cg1) at (0,1) {};
\node[circle,draw=none,fill=Green,inner sep=0pt, minimum size=6pt] (cg2) at (0,0) {};
\draw[thick,Green] (cg1) -- (cg2);
\node[circle,draw=none,fill=Red,inner sep=0pt,minimum size=6pt] (cr1) at (1,1) {};
\node[circle,draw=none,fill=Red,inner sep=0pt,minimum size=6pt] (cr2) at (1,0) {};
\draw[thick,Red] (cr1) -- (cr2);
\end{tikzpicture}
\end{center}
\caption{\label{fig:counterexample3c}}
\end{subfigure}
\caption{Another parallel edges network where full information is suboptimal; it is optimal to reveal no information.}
\end{figure}

 First, we show that full information revelation is not always optimal in parallel edges networks with affine full-state-based (AFS) cost functions.
	
\begin{example}
 \label{ex:counter-stochastic-slope}
	    There are two vertices $V=\{s,t\}$, two parallel edges $E = \{e_1,e_2\}$, and two states $\theta_1$ and $\theta_2$ which appear with prior probability $\prior_{\theta_1} = \prior_{\theta_2} = 1/2$.
        There is a single commodity with a demand of $d = 1$.
	    The cost functions are $c_{e_1}^{\theta_1}(x) = 1$ and $c_{e_2}^{\theta_1}(x) = x$, as well as $c_{e_1}^{\theta_2}(x) = x$ and $c_{e_2}^{\theta_2}(x) = 2$; see Figures~\ref{fig:counterexample2a} and \ref{fig:counterexample2b}.
It is straightforward to verify that for the unique Wardrop equilibrium as a function of the prior $\mu =(\mu_{\theta_1}, \mu_{\theta_2})$, we have
\begin{align*}
x_{e_1}(\mu) &= \min\bigl\{ 2\mu_{\theta_2}, 1\bigr\},  & 
x_{e_2}(\mu) &= \max\bigl\{ 1 - 2\mu_{\theta_2}, 0\bigr\}. 
\end{align*}
The resulting cost of the Wardrop equilibrium is
\begin{align*}
C(\mu) =
\begin{cases}
 2\mu_{\theta_2}^2 - \mu_{\theta_2} + 1 & \text{ if } \mu_{\theta_2} \leq 1/2, \\
 1 & \text{ otherwise;}
\end{cases}
\end{align*}
see Figure~\ref{fig:counterexample2c}.

As a result, the full information revelation signaling scheme $\Phi_{\text{FI}}$ yields a total cost of $C(\Phi_{\text{FI}}) = 1$.
We proceed to compute the optimal public signaling scheme $\Phi_{\text{OPT}}$.
The optimal signal has two signals $\sigma_1$ and $\sigma_2$.
A closer inspection of Figure~\ref{fig:counterexample2c} reveals that we want to choose the signals such that for signal $\sigma_2$ the posterior belief $\mu''$ with $\mu''_{\theta_2} = 1$ is induced and for signal $\sigma_1$ a posterior belief $\mu'$ with $\mu'_{\theta_2} \in (0,1/2)$ is induced.
To this end, let $p \in [0,1/2]$ be a probability to be determined later and consider the following signaling scheme:
\begin{align*}
\phi_{\theta_1,\sigma_1} &= 1/2, & \phi_{\theta_1,\sigma_2} &=  0,  \\
\phi_{\theta_2,\sigma_1} &= p, & \phi_{\theta_2,\sigma_2} &=  1/2-p.  
\end{align*}
When $\sigma_1$ is issued, the induced belief is $\mu' = (\mu'_{\theta_1},\mu'_{\theta_2})$ with $\mu'_{\theta_1} = \frac{1}{2} / (p+ \frac{1}{2})$ and $\mu'_{\theta_2} = p/(p+\frac{1}{2})$ leading to expected cost of
\begin{align*}
C(\mu') = 2 \biggl(\frac{p}{p+\frac{1}{2}}\biggr)^{\!2} - \frac{p}{p+\frac{1}{2}} + 1.
\end{align*}
If, on the other hand, signal $\sigma_2$ is issued, the induced belief is $\mu'' = (\mu''_{\theta_1}, \mu''_{\theta_2})$ with $\mu''_{\theta_1}=0$ and $\mu''_{\theta_2}=1$ with expected cost of $C(\mu'') = 1$.
Finally, multiplying the cost for each signal with the probability of each signal being issued yields the total cost of
\begin{align*}
C(\Phi_{\text{OPT}}) &= \biggl(p \!+\! \frac{1}{2}\biggr) C(\mu') + \biggl(\frac{1}{2} \!-\! p\biggr) C(\mu'') =
\biggl(p + \frac{1}{2}\biggr)\biggl[2 \biggl(\frac{p}{p+\frac{1}{2}}\biggr)^{\!2} - \frac{p}{p+\frac{1}{2}} + 1 \biggr]  + \biggl(\frac{1}{2} - p\biggr) 1.
\end{align*}
This expression is minimized for $p= \frac{1}{2}(\sqrt{2}-1) \approx 0.207$, and we obtain $C(\Phi_{\text{OPT}}) \approx 0.914<1$.
The optimal signaling scheme is indicated in Figure~\ref{fig:counterexample2c}.
In particular, full information revelation is not optimal. \hfill $\blacksquare$
\end{example}

The next example shows that full information revelation need not be optimal in parallel-edge networks with monomial cost functions with state-based offsets (MSO).
	
\begin{example}
\label{ex:counter-monomials}
There are two vertices $V=\{s,t\}$, two parallel edges $E = \{e_1,e_2\}$, and two states $\theta_1$ and $\theta_2$ with prior probability $\prior_{\theta_1} = \prior_{\theta_2} = 1/2$.
There is a single commodity with demand $d = 1$.
The cost functions are $c_{e_1}^{\theta_1}(x) = x^2$ and $c_{e_2}^{\theta_1}(x) = x^2+1$, as well as $c_{e_1}^{\theta_2}(x) = x^2+1$ and $c_{e_2}^{\theta_2}(x) = x^2$; see Figures~\ref{fig:counterexample3a} and \ref{fig:counterexample3b}.

It is straightforward to verify that for the unique Wardrop equilibrium as a function of the belief $\mu = (\mu_{\theta_1}, \mu_{\theta_2})$, we have
\begin{align*}
x_{e_1}(\mu) &= 1 - \mu_{\theta_2}, &   
x_{e_2}(\mu) &= \mu_{\theta_2}.
\end{align*}
The resulting cost of the Wardrop equilibrium is
\begin{align*}
C(\mu) = \mu_{\theta_2}^2 - \mu_{\theta_2} + 1;    
\end{align*}
see Figure~\ref{fig:counterexample3c}.

The full information revelation scheme $\Phi_{\text{FI}}$ has a total expected cost of $\cost(\Phi_{\text{FI}}) = 1$.
	    
From Figure~\ref{fig:counterexample3c}, it is easy to see that the optimal signaling scheme is the no-signaling scheme $\Phi_{\text{NO}}$ that always sends the same signal $\sigma_1$ and results in the commodity splitting the flow equally on edge $e_1$ and $e_2$ resulting in the total expected cost of $\cost(\Phi_{\text{NO}}) = 3/4$. Hence, full information revelation is not optimal.\hfill $\blacksquare$
	\end{example}

As full information revelation may be suboptimal both for AFS and MSO cost functions, it is interesting to analyze if an optimal signaling scheme can be computed in polynomial time. As mentioned above, our techniques do not translate since the costs as a function of the prior cease to be piecewise linear. Moreover, as Example~\ref{ex:counter-stochastic-slope} exhibits, the optimal signaling schemes may require irrational numbers for their description. We believe that the problem of efficiently computing (near)-optimal signals for these settings requires substantially different techniques and leave it as an interesting open problem.

We have further shown that for the case of two states, the optimal signaling scheme can be computed in time that is polynomial in the number of supports that appear in the Wardrop equilibria for all possible beliefs. While our results show that the number of supports may be exponential in the input size of the network, our computational studies exhibited that this parameter is rather low for realistic network instances.

Finally, we have proven that for games on parallel edges with a constant number of states and a constant number of commodities, a cell decomposition approach combined with support enumeration and linear programming techniques yields a polynomial-time algorithm to compute the optimal signaling scheme. It would be interesting to see whether this combination could be applied to further settings, e.g., series-parallel networks where both the number of states and the number of commodities are constant.

We point to the fact that our model considers all possible $s$--$t$ paths in the single-commodity setting. Therefore, the result by \citet{CominettiDS24}, who show that every congestion game is equivalent to routing game on series-parallel networks in which only a subset of paths is allowed, does not generalize our findings to a broader class of congestion games.

\paragraph*{Acknowledgements}
This work was supported by Deutsche Forschungsgemeinschaft EXC-2046/1 (project ID: 390685689) and Ho 3831/9-1 (project ID: 514505843). The authors thank the organizers and participants of Dagstuhl Seminar 18102 ``Dynamic Models in Transportation Science''.

\clearpage
	
\bibliographystyle{informs2014}
\bibliography{information-design}

\clearpage

\section*{Appendix}

\section{Deferred Proofs from \Cref{sec:structural}}
\label{app:structural}

\subsection{Proof of \Cref{lem:piecewise-linear}}
\label{app:lem:piecewise-linear}

\lempiecewiselinear*

\begin{proof}
Let $A\in\mathcal{A}$ be a support such that $P_A\neq \emptyset$ and let $x^*=x^*(\mu)$ for some $\mu\in P_A$ be a Wardrop equilibrium such that $A(x^*)=A$. We proceed to show that $x^*$ and $C$ are affine on $P_A$.
To this end, we define the balance vector $(\beta_v)_{v \in V}$ as
\begin{align*}
\beta_v &= \begin{cases}
 \phantom{-}d & \text{ if } v=s,\\
 -d & \text{ if } v=t,\\
 \phantom{-}0 & \text{ otherwise}
 \end{cases}
 && \text{ for all } v \in V.
\end{align*}
For $v \in V$, let $\delta^+(v)$ and $\delta^-(v)$ be the sets of all out- and ingoing edges of $v$, respectively. By Proposition~\ref{pro:beckmann}, the Wardrop equilibrium $x^*$ is the optimal solution to the optimization problem
\begin{align*}
\text{Min.} \quad &\sum_{e \in E} \int_{0}^{x_e} c_e(z \mid \mu) \;\text{d}z\\
\text{s.t. } \quad &\sum_{e \in \delta^+(v)} x_e - \sum_{e \in \delta^-(v)} x_e = \beta_v
&& \text{ for all } v \in V,\\
&x_e \geq 0 &&\text{ for all } e \in E. 
\end{align*}
By the Karush-Kuhn-Tucker optimality conditions \citep[cf.][Theorems~3.25 and 3.27]{Ruszczyski06}, a flow $x = (x_e)_{e \in E}$ is optimal if and only if it is feasible and there is a dual vector $\pi = (\pi_v)_{v \in V}$ such that

\begin{subequations}
\begin{align}
c_e(x_e \mid \mu) &= \pi_w - \pi_v &&\text{ for all } e = (v,w) \in E \text{ with } x_e \neq 0, \label{eq:def-KKT-1}\\
c_e(x_e \mid \mu) &\geq \pi_w - \pi_v &&\text{ for all } e = (v,w) \in E \text{ with } x_e = 0. \label{eq:def-KKT-2}
\end{align}
\end{subequations}
We claim that in particular the shortest path potential $\psi$ satisfies these conditions. To see this, note that the shortest path potentials fulfill \eqref{eq:def-KKT-2} by definition. Furthermore, any edge $e\in E$ with $x_e\neq 0$ must lie on a shortest $s$-$t$-path and hence \eqref{eq:def-KKT-1} is satisfied as well.
Thus, we may assume that $\pi=\psi$ holds without loss of generality.
Using $c_e(x_e \mid \mu) = a_e x_e + \sum_{\theta \in \Theta} \mu_{\theta} b_e^\theta$, we conclude that a Wardrop equilibrium $x=x(\mu)$ with $\mu\in P_A$ satisfies the following equations
\begin{subequations}
\label{eq:wardrop-equation}
\begin{align}
\pi_v + a_e x_e + \sum_{\theta \in \Theta} \mu_{\theta} b_e^{\theta}  &= \pi_w && \text{ for all } e \in A, \label{eq:wardrop-equation-1}\\
\sum_{e \in \delta^+(v)} x_e - \sum_{e \in \delta^-(v)} x_e &= \beta_v && \text{ for all } v \in V, \label{eq:wardrop-equation-2}\\
\pi_s &= 0, \label{eq:wardrop-equation-3} 
\end{align}
\end{subequations}
as well as the inequalities
\begin{subequations}
\label{eq:wardrop-inequality}
\begin{align}
\pi_v + a_e x_e + \sum_{\theta \in \Theta} \mu_{\theta} b_e^{\theta} &\geq \pi_w && \text{ for all } e \in E \setminus A,\label{eq:wardrop-inequality-1}\\
x_e &\geq 0 && \text{ for all } e \in E.\label{eq:wardrop-inequality-2}
\end{align}
\end{subequations}

We claim that for all $A \in \mathcal{A}$, the linear system \eqref{eq:wardrop-equation-1}--\eqref{eq:wardrop-equation-3} has full rank. To see this, let $\Gamma \in \mathbb{R}^{V \times A}$ be the incidence matrix of the subgraph $G_A$,
i.e., $\Gamma = (\gamma_{v,e})_{v \in V, e \in A}$ defined as $\gamma_{v,e} = 1$, if $e \in \delta^-(v)$, $\gamma_{v,e} = -1$ if $e \in \delta^+(v)$, and $\gamma_{v,e} = 0$, otherwise. Let $D \in \mathbb{R}^{A \times A} = \text{diag}(a_1,\dots,a_{k})$ with $k = |A|$ be the diagonal matrix with the slopes of the cost functions of the edges on the diagonal. Eliminating $\pi_s = 0$, the system \eqref{eq:wardrop-equation} can be written as
\begin{align}
\label{eq:matrix-equation}
\left[
\begin{array}{c c}
D & \hat{\Gamma}^\top\\
\hat{\Gamma} & \mathbf{0}	
\end{array}
\right]
\left[
\begin{array}{c}
x \\
\hat{\pi}
\end{array}
\right] =
\left[
\begin{array}{c}
-\sum_{\theta \in \Theta} \mu_{\theta} b^\theta\\
\hat{\beta} 
\end{array}
\right],
\end{align}
where $\hat{\pi}$ is the vector of vertex potentials with the entry for $s$ removed, $\hat{\Gamma}$ is the incidence matrix with the row for $s$ removed, and $\hat{\beta}$ is the vector $\beta$ with the entry for $s$ removed. Using Schur complements, we obtain that the matrix on the left hand side of \eqref{eq:matrix-equation} is invertible if and only if $\hat{L}:=\hat{\Gamma}D^{-1}\hat{\Gamma}^\top$ is invertible \citep[cf.][Theorem~8.5.11]{Harville97}.
In that case, the inverse is given by
\begin{align} \label{eq:harville}
\left[
\begin{array}{c c}
D & \hat{\Gamma}^\top\\
\hat{\Gamma} & \mathbf{0}	
\end{array}
\right]^{-1} = 
\left[
\begin{array}{c c}
D^{-1} - D^{-1}\hat{\Gamma}^\top \hat{L}^{-1} \hat{\Gamma} D^{-1}\;\;\;  &  \;\;\;D^{-1} \hat{\Gamma}^\top \hat{L}^{-1}\\
\hat{L}^{-1} \hat{\Gamma} D^{-1} & -\hat{L}^{-1}
\end{array}
\right].
\end{align}
In general, the Laplacian matrix $L$ of a connected graph is defined as $L=\Gamma\Gamma^{-1}$, i.e., $\ell_{v,v}$. Thus,
the matrix $\hat{L}$ is a weighted Laplacian matrix of a connected graph (with the entry for $s$ removed) which is known to have full rank. This implies that for fixed $\mu$, there is a unique solution $x$ satisfying \eqref{eq:wardrop-equation} and a unique value for $\pi_t$. 

The equations~\eqref{eq:wardrop-equation} and the inequalities~\eqref{eq:wardrop-inequality} together with $\sum_{\theta \in \Theta} \mu_{\theta} = 1$ and $\mu \geq 0$ define a polytope of all vectors $(\mu,x,\pi) \in \mathbb{R}^{\Theta \times E \times V}$. Here, $\mu \in \Delta(\Theta)$ is a belief, $x$ is a corresponding Wardrop equilibrium with support $A$ where $x_e=0$ for all $e\in E\setminus A$, and $\pi$ is a corresponding vector of vertex potentials. Since the projection of a polytope is a polytope again the set $P_A$ is a polytope as well.

Observe that \eqref{eq:wardrop-equation} gives a system of linear equations which, for fixed $\mu$, has a unique solution in $x$. Thus, the Wardrop equilibrium $x^*(\mu)$ is an affine function in $\mu$ on $P_A$. To see that also the cost of the Wardrop equilibrium is affine on $P_A$ note that the cost of the Wardrop equilibrium is given by $d \pi_t$ and, hence, the result follows.

\end{proof}

\subsection{Proof of \Cref{cor:monotone}}
\label{app:cor:monotone}

\cormonotone*

\begin{proof}
For a fixed support $A \in \mathcal{A}$, we obtain from \eqref{eq:harville} that
\begin{align*}
\left[
\begin{array}{c}
x \\
\hat{\pi}
\end{array}
\right]  = 
\left[
\begin{array}{c c}
D^{-1} - D^{-1}\hat{\Gamma}^\top \hat{L}^{-1} \hat{\Gamma} D^{-1} \;\;\; & \;\;\;  D^{-1} \hat{\Gamma}^\top \hat{L}^{-1}\\
\hat{L}^{-1} \hat{\Gamma} D^{-1} & -\hat{L}^{-1}
\end{array}
\right]
\left[
\begin{array}{c}
-\sum_{\theta \in \Theta} \mu_{\theta} b^\theta\\
\hat{\beta} 
\end{array}
\right],
\end{align*}
where $\hat{\beta}_t = -d$ and $\hat{\beta}_v = 0$, otherwise. We obtain
\begin{align*}
\hat{\pi} = -\hat{L}^{-1} \hat{\Gamma}D^{-1}\Biggl(\sum_{\theta \in \Theta} \mu_{\theta} b^{\theta}\Biggr)  - \hat{L}^{-1}\hat{\beta},
\intertext{and, in particular,}
\pi_t = -e_t \Biggl(\hat{L}^{-1} \hat{\Gamma}D^{-1} \Biggl(\sum_{\theta \in \Theta} \mu_{\theta} b^{\theta}\Biggr) \Biggr)   + d \cdot \ell_{t,t},
\end{align*}
where $e_t$ is the unit vector corresponding to row $t$ and $\ell_{t,t}$ is the diagonal entry of $\hat{L}^{-1}$ that belongs to row and column $t$.
It was shown in \cite[Lemma 2.5]{warode2022parametric} that $\hat{L}$ is an $M$-matrix, i.e., all off-diagonal entries are non-positive, and the real parts of its eigenvalues are non-negative. Since $M$-matrices are inverse-positive, we have $\ell_{t,t} \geq 0$. This shows that $\pi_t$ is non-decreasing in $d$. To show that $\pi_t$ is strictly increasing in $d$, we have to invest some more effort by essentially revisiting the proof that inverses of $M$-matrices are non-negative (with some minor tweaks).

Since $\hat{L}$ is an $M$-matrix, we can write it as $\hat{L} = sI - \Matri$ with $s > 0$, $I$ the identity matrix, and $\Matri = (\matri_{i,j})$ a non-negative matrix.  
By making $s$ large enough, we can ensure that $\matri_{t,t} > 0$.
By the Perron-Frobenius-Theorem, the spectral radius $\rho(\Matri)$ of $\Matri$ is attained for a non-negative eigenvalue, i.e., there is a non-negative eigenvalue $\lambda^*$ of $\Matri$ such that $\rho(\Matri) = \lambda^*$.
As $\hat{L}$ is a symmetric $M$-matrix, all eigenvalues of $\hat{L}$ are positive and real.
Further, for any eigenvalue $\lambda$ of $\Matri$, we have that $s- \lambda$ is an eigenvalue of $\hat{L}$. In particular, we have that $s - \rho(\Matri)$ is an eigenvalue of $\hat{L}$. Since $\hat{L}$ only has non-negative eigenvalues, we have that $s - \rho(\Matri) > 0$, i.e., $\rho(\Matri) < s$. Let $\Matri' = \frac{1}{s}\Matri$, then $\rho(\Matri') < 1$. 
Let $\hat{L}' = \frac{1}{s}\hat{L}$. Using that $\rho(\Matri')< 1$, the series $\sum_{k=0}^n (\Matri')^k$ converges and satisfies the equation
\begin{align*}
\hat{L}' \sum_{k=0}^\infty (\Matri')^k = (I - \Matri') \sum_{k=0}^\infty (\Matri')^k = \Biggl(\sum_{k=0}^\infty (\Matri')^k\Biggr) - \Biggl(\sum_{k=1}^\infty (\Matri')^k \Biggr) = (\Matri')^0 = I.
\end{align*}
We have established
\begin{align*}
 \hat{L}^{-1} = (s\hat{L}')^{-1} = \frac{1}{s} (\hat{L}')^{-1} = \frac{1}{s} \sum_{k=0}^\infty (\Matri')^k .   
\end{align*}
We have chosen $s$ such that $\matri_{t,t} > 0$ and hence also $\matri'_{t,t} = \matri_{t,t}/s > 0$. Since $\Matri'$ is non-negative, for the higher powers of $(\Matri')^k$ with $k \geq 1$ only non-negative terms are added to $\matri'_{t,t}$. We thus obtain that $\ell_{t,t} \geq \matri'_{t,t}/s > 0$, as claimed.

This shows that, for a fixed support, $\pi_t$ is strictly increasing in $d$.
Since $\pi_t$ is continuous in~$d$ \cite[Proposition~3.1]{CominettiDS24}, the result follows.

\end{proof}

\section{Deferred Proofs from \Cref{sec:sepa}}
\label{app:sepa}

\subsection{Proof of \Cref{lem:smaller-support}}
\label{app:tlem:smaller-support}

\lemsmallersupport*

\begin{proof}
		We prove the statement by induction over $|E|$.
		
		For $|E|=1$ the graph has only a single edge~$e = \{s,t\}$. The only support $A \in \mathcal{A}$ is $A = E$. Clearly, $P_E = \Delta(\Theta)$ and the statement holds trivially since there is no $\mu \in \Delta(\Theta) \setminus P_E$.
		
		Fix $k \in \mathbb{N}$ and suppose that the statement of the lemma holds for all series-parallel graphs with up to $k$ edges. Consider a series-parallel graph $G = (V,E)$ with $k+1$ edges.
  Since $G$ is series-parallel, there is a sequence of serial and parallel compositions of smaller series-parallel graphs that ends in $G$. In particular, $G$ is constructed either by a final serial composition of two smaller series-parallel graphs $G_1 = (V_1,E_1)$ and $G_2 = (V_2, E_2)$, or by a final parallel composition of $G_1$ and $G_2$. Let $A_1 = A \cap E_1$ and $A_2 = A \cap E_2$.
For $j \in \{1,2\}$, we denote by $\mathcal{A}_j$ the set of supports for $G_j$ such that the corresponding subgraph is connected and spans $V_j$.
		Further, for a support $T \in  \mathcal{A}_j$, let $C_T^{j}(\cdot)$ denote the cost of the solution of the linear system \eqref{eq:wardrop-equation} with support $T$ for $G_j$.
		We proceed to distinguish the following two cases.
				
		\paragraph{First case: Final composition is serial.}
		Let $A \in \mathcal{A}$ with $\mu \in \Delta(\Theta) \setminus P_A$ be given.
Since the final composition is serial, we have $A_j \in \mathcal{A}_j$ for all $j \in \{1,2\}$ where at least one of $A_1$ and $A_2$ is infeasible. However, both $G_1$ and $G_2$ have at most $k$ edges, so we can apply the induction hypothesis on either of these graphs and obtain the support $A_j' \in \mathcal{A}_j$ such that either $C^{j}_{A_j'}(\mu) < C^{j}_{A_j}(\mu)$ or $C^{j}_{A_j'}(\mu) = C^{j}_{A_j}(\mu)$ and $A_j'$  is feasible for all $j\in \{1,2\}$.    
 		Since $G_1$ and $G_2$ were composed in series, we have for any support $T \in \mathcal{A}$ that
		\begin{align*}
		C_T(\mu)=C^{1}_{T \cap E_1}(\mu) + C^{2}_{T \cap E_2}(\mu).
		\end{align*}
		We set $A' = A_1' \cup A_2' \in \mathcal{A}$.
        Since $A_j'$ is a support of $G_j$ for all $j\in \{1,2\}$ and the composition is serial, it is easy to verify that $A'$ is a support for $G$, i.e., $G= (V,A')$ is connected and contains an $(s,v)$-path for all $v \in V$.
        Furthermore, if $A_j'$ is feasible for all $j\in \{1,2\}$, $A'$ is also feasible for $G$.
        On the other hand, if at least one of $A_1'$ and $A_2'$ is not feasible, then $C^{j}_{A_j'}(\mu) < C^{j}_{A_j}(\mu)$ for at least one of the two supports while for the other one we have $C^{j}_{A_j'}(\mu) \leq C^{j}_{A_j}(\mu)$.
        Hence, we then obtain
		\begin{align*}
		C_{A'}(\mu)=C^{1}_{A_1'}(\mu) + C^{2}_{A_2'}(\mu) < C^{1}_{A_1}(\mu) + C^{2}_{A_2}(\mu) = \cost_A(\mu),
		\end{align*}
        so we have found a new support $A' \in \mathcal{A}$ with strictly less cost than $A$. 
				
		\paragraph{Second case: Final composition is parallel.}
		Let again $A \in \mathcal{A}$ with $\mu \in \Delta(\Theta) \setminus P_A$ be given.
 A major complication compared to the first case is that we cannot assume anymore that $A_j \in \mathcal{A}_j$ for all $j \in \{1,2\}$ since only one of the two parallel components may contain an $s$-$t$-path.
		
		For $j \in \{1,2\}$, let $\lambda_j = \sum_{e \in \delta^+(s) \cap E_j} (x^*_A)_e$ be the total flow in $G$ send over the parallel component $G_j$. We have $\lambda_1 + \lambda_2 = d$.
        It is without loss of generality to assume that $\lambda_1\leq \lambda_2$.
		
		\paragraph{Subcase: $\lambda_1 < 0$.} With $\lambda_1 < 0$, we have $A_1 \in \mathcal{A}_1$, since there is a non-zero flow in that component, and, therefore, $t$ can be reached from $s$.
		Let $A_1' = A_1 \setminus (\delta^-(t) \cap E_1)$. Since $\lambda_1 < 0$, we have $\lambda_2 > d$ and, thus, there is a path from $s$ to $t$ in $A_2$ and, in particular, $A' =A_1' \cup A_2 \in \mathcal{A}$.
		
		In the following, we write $\lambda_j'$ and $\pi'_v$ for the values of $\lambda_j$ and $\pi_v$ for the new support $A'$.
		Then, we have $\lambda_1'= 0$ and, hence $\lambda_2' = d$.
		By the proof of Corollary~\ref{cor:monotone}, the per-unit cost $\pi_t$ of an equilibrium flow (i.e., a flow that satisfies the linear system \eqref{eq:wardrop-equation} for a fixed support) is strictly increasing in $\lambda$. 
		Hence $\pi_t' < \pi_t$. With $C_{A'}(\mu) = \pi_t' d < \pi_t d = C_A(\mu)$, the result follows.
		
		\paragraph{Subcase: $\lambda_1 = 0$.} 
		In this subcase, a potential issue is that $A_1$ may not be contained in $\mathcal{A}_1$ since there may not be a path from $s$ to $t$ in $A_1$. As a consequence, we may not be able to apply the induction hypothesis on $G_1$. However, with $\lambda_2=d$, we know that $A_2$ is a support for $G_2$.
  		We first compute a shortest path tree with respect to $c_e(0 \mid \mu)$ in $G_1$ and obtain a shortest path potential $\pi_v'$ for all $v\in V_1$.
        
        If $\pi_t> \pi_t'$, the $s$-$t$-path in the shortest path tree of $G_1$ has lower costs than the used $s$-$t$-paths in $G_2$.
        We set $A_1'=\{e=\{v,w\} \in E_1\mid c_e(0\mid \mu)=\pi_w'-\pi_v' \}$ and $A'=A_1'\cup A_2$.
        Hence, in equilibrium, we obtain new flow values $\lambda_1'>0$ and $\lambda_2'<d$. The result then follows from the monotonicity (proof of Corollary~\ref{cor:monotone}) which causes the per-unit cost in $G_2$ to decrease and therefore $C_{A'}(\mu) < C_{A}(\mu)$. 
        
        If $\pi_t\leq \pi_t'$, we set $A_1'=\{e=\{v,w\} \in E_1\mid c_e(0 \mid \mu)=\pi_w'-\pi_v' \}\setminus \{e_{st}\}$, where $e_{st}$ is the last edge on the unique $s$-$t$-path in that tree.
        Note, that $A_1'$ is not a support for $G_1$ since vertex $t$ is not reached from $s$, and in addition, there may be vertices, which used to be connected to $s$ by a path via $t$ in the shortest path tree.

        If $A_2$ is a feasible support for $G_2$, it follows that $A'=A_1'\cup A_2$ is a feasible support of $G$ with $C_A(\mu)=C_{A'}(\mu)$ and there is nothing left to show.
        
        If on the other hand, $A_2$ is not a feasible support for $G_2$, we apply the induction hypothesis on $G_2$ and obtain a new support $A_2'$ of $G_2$ that is either feasible or has lower cost on $G_2$ than $A_2$. In the first case, we use the same argumentation as above, where we assumed that $A_2$ was feasible and are done.
        In the latter case, we set $A'=A_1\cup A_2'$. Since there is no $s$-$t$-path in $A_1$ and therefore still no flow on $G_1$, all flow is on $G_2$ using only edges of $A_2'$. Hence, $C_{A'}(\mu)=C^{2}_{A_2'}(\mu) < C^{2}_{A_2}(\mu) = \cost_A(\mu)$, so we have found a support of $G$ with strictly less cost than $A$.
		
	\paragraph{Subcase: $\lambda_1 > 0$.} In this case, both subgraphs $G_1$ and $G_2$ carry flow which implies that $A_j \in\mathcal{A}_j$ for all $j\in\{1,2\}$. Applying the induction hypothesis on $G_1$ and $G_2$, we end up in a similar situation as in the case where the last composition was serial. We either obtain two supports $A_1'$ and $A_2'$ that are both feasible and $C_{A'_j}(\mu) = C_{A_j}(\mu)$ for all $j\in\{1,2\}$, or for at least one of the supports $A_1'$ and $A_2'$, we have $C_{A'_j}(\mu) < C_{A_j}(\mu)$. In the first case, it is easy to see that $A'=A_1'\cup A_2'$ is a feasible support for $G$, that has the same cost as $A$ and the same flow distributions $\lambda_1$ and $\lambda_2$.
    In the latter case, the result follows from the monotonicity of the flows since in equilibrium less flow will be sent via the more expensive subgraph and, thus, the per-unit cost in both subgraphs decreases to an equal level which is smaller than the per-unit cost in $A$.
		
\end{proof}

\subsection{Proof of \Cref{thrm:sepa-full-opt-iff}}
\label{app:thrm:sepa-full-opt-iff}

\thrmsepafulloptiff*

\begin{proof}
    The if-part follows from Theorem~\ref{thrm:sepa-full-opt}.
    To prove the only-if part, it suffices to show that for a non-series-parallel graph there exist cost functions $c^\theta_e :\mathbb{R}_{\geq 0}\rightarrow \mathbb{R}$, $e\in E,\theta \in \Theta$ such that full information revelation is not optimal.
    We call a graph with two designated vertices $s,t \in V$ a two-terminal graph.
    In the following we make use of some definitions by \cite{Duffin65}.
    We call two edges $e,e'\in E$ \emph{confluent} if there are no
    two (undirected) simple cycles $C_1$ and $C_2$ both containing $e$ and $e'$ such that the two edges have the same orientation in $C_1$ and a different orientation in $C_2$.
    Further, an edge is \emph{$s$-$t$-confluent} if it is confluent with the (virtually added) edge $(t,s)$.
    As shown by Duffin, a two-terminal graph $G$ is series-parallel if and only if all edges are $s$-$t$-confluent.	
    Let $G=(V,E)$ be a two-terminal graph with source $s\in V$ and sink $t\in V$ such that $G$ is not series-parallel, i.e., there exists an edge $b\in E$ that is not confluent with the (virtually added) edge $a:=(t,s)$.
    Hence, there exist two cycles $C_1$ and $C_2$ containing $a$ and $b$ such that $a$ is used in the same direction in both cycles but the direction of $b$ changes.
    For an illustration, see cycle $C_1=(s,v_h,v_i,v_j,v_{j+1},v_k,v_l,t,s)$ and $C_2=(s,v_h,v_k,v_{j+1},v_j,v_i,v_l,t,s)$ in \Cref{fig:wheatstone}.
    Note that the paths represented by a dotted line may consist of an arbitrary number of edges (including $0$, in which case the corresponding vertices are the same).
    Any dashed path contains at least one edge.
    Next, we choose an arbitrary edge on each dashed path and label it $e_{hi},e_{hk},e_{il},$ and $e_{kl}$, respectively.
    We now define the cost function as follows:
    \begin{align*}
        c_e^{\theta_1}(x)=
        \begin{cases}
            x & \text{if } e\in \{e_{hi},e_{kl}\},\\
            1 & \text{if } e\in \{e_{hk},e_{il}\},\\
            \infty & \text{if } e\in E\setminus (E[C_1]\cup E[C_2]),\\
            0 & \text{otherwise},
        \end{cases}
        &&
        c_e^{\theta_2}(x)=
        \begin{cases}
            1 & \text{if } e=b,\\
            c_e^{\theta_1}(x) &\text{otherwise}.
        \end{cases}
    \end{align*}
    Note that only the cost of edge $b$ depends on the state.
    Ignoring all edges that have cost either 0 or $\infty$ in both states, we obtain the embedded graph shown in Figure~\ref{fig:wheatstone-net}.
    With the cost function defined above, we obtain the example illustrated in Figure~\ref{fig:braess} for which we showed that full information revelation is not an optimal solution.

\end{proof}
	
\begin{figure}[t!]
\scriptsize
\begin{subfigure}[t]{0.5\textwidth}
\begin{center}
			\begin{tikzpicture}[xscale=0.8,yscale=0.8,shorten > = 0pt]
				> = stealth, 
				shorten > = 0pt, 
				auto,
				node distance = 3cm, 
				thick, 
				scale = 0.8
				]
				
				\tikzstyle{every state}=[
				draw = black,
				thick,
				fill = white,
				inner sep=0pt,
				minimum size = 7mm,
				]
				
				\node[state,label=left:{$s$}] (s) at (-2,0) [circle] {}; 
				\node[state,label=right:{$t$}] (t) at (6,0) [circle] {}; 
				\node[state,label=above:{$v_i$}] (v1) at (2,2) [circle]{}; 
				\node[state,label=below:{$v_k$}] (v2) at (2,-2) [circle]{}; 
				\node[state,label=below left:{$v_h$}] (vh) at (0,0) [circle] {};
				\node[state,label=below right:{$v_l$}] (vl) at (4,0) [circle] {};
				\node[state,label=right:{$v_j$}] (vj) at (2,2/3) [circle]{}; 
				\node[state,label=right:{$v_{j+1}$}] (vj1) at (2,-2/3) [circle]{};

				\path[] (vh) edge node[above left] {$e_{hi}$} (v1)[dashed,thick];
				\path[] (vh) edge node[below left] {$e_{hk}$} (v2)[dashed,thick];
				\path[] (v1) edge node[above right] {$e_{il}$} (vl)[dashed,thick];
				\path[] (v2) edge node[below right] {$e_{kl}$} (vl)[dashed,thick];
				\path[] (v1) edge (vj)[dotted,thick];
				\path[] (vj) edge node[right] {$b$} (vj1);
				\path[] (vj1) edge (v2)[dotted,thick];
				\path[] (s) edge (vh)[dotted,thick];
				\path[] (vl) edge (t)[dotted,thick];		
				\draw[thick] (s) ..controls (-2,4) and (6,4)..  node[above,yshift=0mm] {$a$} (t); 
			\end{tikzpicture}
			\end{center}
			\caption{}
			\label{fig:wheatstone}
		\end{subfigure}\hfill
		\begin{subfigure}[t]{0.5\textwidth}
			\begin{center}
			\begin{tikzpicture}[xscale=0.9,yscale=0.9,shorten > = 0pt]
				node distance = 3cm, 
				thick, 
				scale = 0.8
				]
				
				\tikzstyle{every state}=[
				draw = black,
				thick,
				fill = white,
				inner sep=0pt,
				]
				
				\node[state,label=left:{$s$}] (s) at (0,0) [circle] {}; 
				\node[state,label=right:{$t$}] (t) at (4,0) [circle] {}; 
				\node[state,label=above:{$v_i$}] (v1) at (2,2) [circle]{}; 
				\node[state,label=below:{$v_k$}] (v2) at (2,-2) [circle]{}; 
				
				\path[] (s) edge node[above left] {$e_{hi}$} (v1);
				\path[] (s) edge node[below left] {$e_{hk}$} (v2);
				\path[] (v1) edge node[above right] {$e_{il}$} (t);
				\path[] (v2) edge node[below right] {$e_{kl}$} (t);
				\path[] (v1) edge node[right] {$b$} (v2);		
			\end{tikzpicture}
			\end{center}
			\caption{}
			\label{fig:wheatstone-net}
		\end{subfigure}
		\caption{Illustrations for the proof of Theorem~\ref{thrm:sepa-full-opt-iff}: (a) cycles $C_1$ and $C_2$ where edges are bold and paths are dashed; (b) the embedded Braess graph with $s=v_h$ and $t=v_l$.}
\end{figure}

\section{Deferred Proofs from \Cref{sec:LPs}}
\label{app:LPs}

\subsection{Proof of \Cref{thm:bigLP}}
\label{app:thm:bigLP}

\thmbigLP*

\begin{proof}
    First, consider a single signal $\sigma \in [k]$. In the emerging Wardrop equilibrium, every commodity $i$ uses only cost-optimal $s_i$-$t_i$-paths. For any belief $\mu_\sigma$ that results in a Wardrop equilibrium with fixed supports $(\support_{i,\sigma})_{i \in R}$, we can extend the description of Wardrop flows with vertex potentials and flow-conservation constraints developed in Lemma~\ref{lem:piecewise-linear}, as precised in the following. A straightforward adaptation of the Karush-Kuhn-Tucker conditions to multi-commodity games allows to generalize the polytope $P_A$ described by~\eqref{eq:wardrop-equation} and \eqref{eq:wardrop-inequality} to multi-commodity games as follows. For each $i \in R$, we use the balance vector $(\beta_{v,i})_{v \in V}$ 
    \[
        \beta_{v,i} = \begin{cases} \phantom{-}d_i & \text{ if } v = s_i,\\ -d_i & \text{ if } v = t_i, \text{ and} \\ \phantom{-}0 & \text{ otherwise,} \end{cases}
    \]
    to define the system of inequalities
    \begin{align}
        \label{eq:wardrop-inequality-multi}
        \begin{aligned}
            \pi_{v,i,\sigma} + a_e x_{e,\sigma} + \sum_{\theta \in \Theta} \mu_{\theta,\sigma} b_e^{\theta}  &= \pi_{w,i,\sigma} && \text{ for all } e=(v,w) \in \support_{i,\sigma}, i \in R,\\
            \pi_{v,i,\sigma} + a_e x_{e,\sigma} + \sum_{\theta \in \Theta} \mu_{\theta,\sigma} b_e^{\theta} &\geq \pi_{w,i,\sigma} && \text{ for all } e=(v,w) \in E \setminus \support_{i,\sigma}, i \in R,\\
            \sum_{e \in \delta^+(v)} x_{e,i,\sigma} - \sum_{e \in \delta^-(v)} x_{e,i,\sigma} &= \beta_{v,i} && \text{ for all } v \in V, i \in R,\\
            x_{e,\sigma} &= \sum_{i \in R} x_{e,i,\sigma} && \text{ for all } e \in E,\\ 
            x_{e,i,\sigma} &\geq 0 && \text{ for all } e \in E, i \in R, \\
            \pi_{s_i,i,\sigma} &= 0 \,.
        \end{aligned}
    \end{align}
    Constraints~\eqref{eq:wardrop-inequality-multi} capture the Wardrop flow with given supports $(\support_{i,\sigma})_{i \in R}$. Let us turn to the conditional belief $\mu_\sigma$ over states. It emerges from the signaling probabilities $\phi_{\theta,\sigma}$ and is described by the following constraints 
	\begin{align}
	\label{eq:condDist}
	\begin{aligned}
		\phi_{\theta,\sigma} &\le \prior_\theta  && \text{for all } \theta \in \Theta,\\
		\phi_{\theta,\sigma} &\ge 0  && \text{for all } \theta\in \Theta,\\
		\phi_{\theta,\sigma} &= \D \phi_\sigma \cdot \mu_{\theta,\sigma} && \text{for all } \theta\in \Theta,\\
		\phi_{\sigma} &= \D \sum_{\theta \in \Theta} \phi_{\theta,\sigma}.
	\end{aligned}
	\end{align}
	If signal $\sigma$ is not issued, then $\phi_{\sigma} = 0$. Hence, $\phi_{\theta,\sigma} = 0$ for all $\theta\in\Theta$ and the constraints \eqref{eq:wardrop-inequality-multi} and~\eqref{eq:condDist} are not meaningful since no conditional belief $\mu_\sigma$ is formed and hence no supports $(A_{i,\sigma})_{i \in R}$ exist. Instead, suppose signal $\sigma$ is issued.
	Then, \eqref{eq:wardrop-inequality-multi} and~\eqref{eq:condDist} describe the polytope of conditional beliefs that result in a Wardrop equilibrium on supports $\support_{i,\sigma}$.
	
	The constraint $\phi_{\theta,\sigma} = \phi_\sigma \cdot \mu_{\theta,\sigma}$ in \eqref{eq:condDist} is non-linear.
	We substitute $\mu_{\theta,\sigma} = \phi_{\theta,\sigma}/\phi_{\sigma}$ in~\eqref{eq:wardrop-inequality-multi} and multiply all (in-)equalities of~\eqref{eq:wardrop-inequality-multi} by $\phi_{\sigma} > 0$.
	Afterward, we further substitute $y_{e,i,\sigma} = x_{e,i,\sigma} \cdot \phi_{\sigma}$ and $\tau_{v,i,\sigma} = \pi_{v,i,\sigma} \cdot \phi_{\sigma}$.
	Being a mixture of demand flow and signal probability, $y_{e,i,\sigma}$ can be interpreted as the flow of ``probabilistic demand'', whereas $\tau$ becomes the standard vertex potential for this flow. The nonlinear constraint $\mu_{\theta,\sigma} = \phi_{\theta,\sigma}/\phi_{\sigma}$ as well as variables $\mu_{\theta,\sigma}$ and $\phi_\sigma$ can be omitted. This yields the following polytope that is equivalent to \eqref{eq:wardrop-inequality-multi}$+$\eqref{eq:condDist}:
	\begin{equation}
	\label{eq:signalSupportNetwork}
	\begin{array}{rcll}
		\D \tau_{v,i,\sigma} + a_e y_{e,\sigma} + \sum_{\theta \in \Theta} \phi_{\theta,\sigma} b_e^\theta & = & \tau_{w,i,\sigma} & \text{for all } e=(v,w) \in \support_{i,\sigma}, i \in R, \\
		\D \tau_{v,i,\sigma} + a_e y_{e,\sigma} + \sum_{\theta \in \Theta} \phi_{\theta,\sigma} b_e^\theta & \ge & \tau_{w,i,\sigma} & \text{for all } e=(v,w) \in E \setminus \support_{i,\sigma}, i \in R, \\
		\D \sum_{e \in \delta^+(v)} y_{e,i,\sigma} - \sum_{e \in \delta^{-}(v)} y_{e,i,\sigma} &=& \beta_{i,v} \cdot \D \sum_{\theta \in \Theta} \phi_{\theta,\sigma} \phantom{~~~} & \text{for all } v \in V, i\in R,\\
		y_{e,\sigma} &=& \D \sum_{i \in R} y_{e,i,\sigma} & \text{for all } e\in E,\\
		y_{e,i,\sigma} & \ge & 0  & \text{for all } e \in \support_{i,\sigma}, i \in R, \\
	    \tau_{s_i,i,\sigma} & = & 0 & \text{for all } i \in R, \\
		\phi_{\theta,\sigma} & \le & \prior_\theta  & \text{for all } \theta\in \Theta,\\
		\phi_{\theta,\sigma} & \ge & 0  & \text{for all } \theta\in \Theta.
	\end{array}
	\end{equation}
	For every signal $\sigma$ this polytope has a trivial all-zero solution (i.e., $\tau_{v,i,\sigma}=0$ for all $v\in V,i\in R$, $y_{e,\sigma}=0$ for all $e\in E$, and $\phi_{\theta,\sigma}=0$ for all $\theta\in \Theta$) which can be interpreted as the signal not being issued.
	Every non-zero solution corresponds to (a part of) a signaling scheme $\Phi$ that includes the signal $\sigma$ with a resulting Wardrop equilibrium using the given support $\support_{i,\sigma}$. 
 	We can describe the set of all schemes $\Phi$ by combining all polytopes for individual signals $\sigma$ with given supports $\support_{i,\sigma}$ in~\eqref{eq:signalSupportNetwork} and adding the decomposition constraint for the prior $\prior$
    \begin{align}
	   \label{eq:dist}
		\sum_{\sigma \in [k]} \phi_{\theta,\sigma} &= \prior_\theta  && \text{for all } \theta\in \Theta.
	\end{align}
    We intend to optimize over this polytope of signaling schemes, i.e., we strive to find a scheme with smallest total expected cost $C(\Phi)$. The cost can be expressed as follows
	\[	
        \cost(\Phi) = \sum_{\sigma \in [k]}\phi_\sigma \cdot \cost(\mu_\sigma) = \sum_{\sigma \in [k]} \phi_{\sigma} \cdot \left(\sum_{i \in R} d_i \cdot \pi_{t_i,i,\sigma}\right) = \sum_{\sigma \in [k]} \sum_{i \in R} d_i \cdot \tau_{t_i,i,\sigma},
	\]
	i.e., the total cost of a signaling scheme is equal to the weighted sum of potentials $\tau_{t_i,i,\sigma}x$ at the destination $t_i$ for all signals $\sigma \in \Sigma$. We here used $\tau_{e,i,\sigma} = \pi_{e,i,\sigma} \cdot \phi_{\sigma}$ as defined above. As a consequence, finding an optimal signaling scheme for a given set of support vectors can be formulated as the following linear program
	
	\begin{equation}
		\label{eq:schemeSupportNetwork}
		\begin{array}{lrcll}
			\text{Min. } & \multicolumn{3}{l}{\D \sum_{\sigma \in [k]} \sum_{i \in R} d_i \cdot \tau_{t_i,i,\sigma} } \\
			\text{s.t. } &\D \tau_{v,i,\sigma} + a_e y_{e,\sigma} + \sum_{\theta \in \Theta} \phi_{\theta,\sigma} b_e^\theta & = & \tau_{w,i,\sigma} & \text{for all } e=(v,w) \in \support_{i,\sigma}, i \in R, \sigma \in [k], \\
			&\D \tau_{v,i,\sigma} + a_e y_{e,\sigma} + \sum_{\theta \in \Theta} \phi_{\theta,\sigma} b_e^\theta & \ge & \tau_{w,i,\sigma} & \text{for all } e=(v,w) \in E \setminus \support_{i,\sigma}, i \in R, \sigma \in [k],\\
	   	    &\D \sum_{e \in \delta^+(v)} y_{e,i,\sigma} - \sum_{e \in \delta^{-}(v)} y_{e,i,\sigma} &=& \beta_{i,v} \cdot \D \sum_{\theta \in \Theta} \phi_{\theta,\sigma} \phantom{~~~}& \text{for all } v \in V, i\in R,\\
			&y_{e,\sigma} &=& \D \sum_{i \in R} y_{e,i,\sigma} & \text{for all } e\in E, \sigma \in [k],\\
			&y_{e,i,\sigma} & \ge & 0  & \text{for all } e \in \support_i, i \in R, \sigma \in [k], \\
		    &\tau_{s_i,i,\sigma} & = & 0 & \text{for all } i \in R, \sigma \in [k], \\
			&\D \sum_{\sigma \in [k]} \phi_{\theta,\sigma} & = & \prior_\theta & \text{for all } \theta\in \Theta,\\
			&\phi_{\theta,\sigma} & \le & \prior_\theta & \text{for all } \theta\in \Theta, \sigma \in [k],\\
			&\phi_{\theta,\sigma} & \ge & 0  & \text{for all } \theta\in \Theta, \sigma \in [k].
		\end{array}
	\end{equation}
    The number of variables and constraints is a polynomial in $|\Theta|$, $|E|$, $|R|$, and $k$. As such, the LP can be solved in polynomial time.
    
\end{proof}

\subsection{Proof of \Cref{thm:expoential-supports}}
\label{app:thm:expoential-supports}

\thmexpoentialsupports*

\begin{figure}[t!]
 \scriptsize
 
	\begin{center}
			\begin{tikzpicture}[
			> = stealth, 
			shorten > = 1pt, 
			auto,
			node distance = 3cm, 
			thick, 
			xscale = 0.47,
			yscale=0.67
			]
			
			\tikzstyle{every state}=[
			draw = black,
			thick,
			fill = white,
			minimum size = 4mm
			]
			
			\node[state,label=left:{\normalsize $s$}] (s) at (0,0) [circle] {}; 
			\node[state,label=right:{\normalsize $t$}] (t) at (12,0) [circle] {}; 
			\node[state,label=above:{\normalsize $v_6$}] (v6) at (6,3) [circle] {}; 
			\node[state,label=below:{\normalsize $v_1$}] (v1) at (6,-3) [circle] {}; 
			\node[state,label=left:{\normalsize $v_2$}] (v2) at (2.5,0) [circle] {}; 
			\node[state,label=below:{\normalsize $v_3$}] (v3) at (6,-1) [circle] {}; 
			\node[state,label=above:{\normalsize $v_4$}] (v4) at (6,1) [circle] {}; 
			\node[state,label=right:{\normalsize $v_5$}] (v5) at (9.5,0) [circle] {};

			\draw[->] (s) ..controls (0,6) and (12,6)..  node[above,yshift=0mm] {$ 300 \alpha $} (t); 
			\path[->] (s) edge node[above,yshift=0mm, rotate=-325] {$100$} (v6); 
			\path[->] (v6) edge node[above,yshift=0mm, rotate=-35] {$200x$} (t); 
			\path[->] (s) edge node[below,yshift=-0mm,rotate=325] {$200x$} (v1); 
			\path[->] (v1) edge node[below,yshift=0mm,rotate=35] {$100$} (t) [];
			\path[->] (v1) edge node[above,pos=.4,xshift=0mm,yshift=-0mm, rotate=-52] {$200x$} (v2) []; 
			\path[->] (v1) edge node[above,pos=.4, xshift=0mm,yshift=0mm,rotate=-308] {$10$} (v5) [];  
			\path[->] (v2) edge node[above,pos=.2,xshift=0mm, yshift=0mm,rotate=52] {$10$} (v6) [];
			\path[->] (v5) edge node[above,pos=.2,xshift=0mm,yshift=0mm,rotate=308] {$200x$} (v6) [];  
			\path[->] (v2) edge node[above,pos=.6,yshift=-0.0mm,rotate=-340] {$1$} (v4) [];    
			\path[->] (v2) edge node[above,pos=.6,yshift=-0mm,rotate=-20] {$200x$} (v3) [];
			\path[->] (v4) edge node[above,pos=.4,yshift=-0mm,rotate=340] {$200x$} (v5) [];
			\path[->] (v3) edge node[above,pos=.4,yshift=-0.0mm,rotate=20] {$1$} (v5) [];  
            \path[->] (v3) edge node[xshift=0.2mm] {$0$} (v4) [];  
			\end{tikzpicture}
		\end{center}	
		\caption{Illustration of the modified nested Braess graph $G^{B+p}_3$ with expected edge costs given belief $\mu_\alpha := (1-\alpha, \alpha), \alpha \in [0,1]$. }
        \label{fig:mixed_braess_j=3}
	\end{figure}

\begin{proof}
    The class of games is based on the family of nested Braess graphs considered by \citet[Section 6.1.2]{KlimmW22}. For $j \in \NN$, they show that there are $2^{j+1}$ different demand rates $d_1,...,d_{2^{j+1}} \in [0,3\cdot 10^{j-1}]$ such that all corresponding supports $A^1,...,A^{2^{j+1}}$ are different for the $j$th graph in their construction.
    
  We adopt their construction to our setting with state-based cost functions and unit demand $d=1$ as follows. For given $j \in \NN$, we define the \emph{$j$th nested Braess graph} $G^B_j =(V_j,E_j)$ with vertex set $V_j := \lbrace v_0, v_1, v_2,...,v_{2j}, v_{2j+1} \rbrace$\footnote{We denote $v_0 := a$ and $ v_{2j+1} := t$ for the source and target vertex, respectively.} and edge set $E_j := E^1_j \cup E^2_j \cup E^3_j \cup \lbrace (v_j,v_{j+1}) \rbrace$, where
  \begin{align*}
  E^1_j &:= \lbrace (v_i,v_{i+1}) \mid i \in \lbrace 0,...,2j \rbrace, i \not = j \rbrace, \\
  E^2_j &:= \lbrace (v_i,v_{2j-i}) \mid i \in \lbrace 0,...,j-1 \rbrace, \\
  E^3_j &:= \lbrace (v_{i+1},v_{2j+1-i}) \mid i \in \lbrace 0,...,j-1 \rbrace \rbrace.
  \end{align*}
   We also define $E^O_j := \lbrace \left( s, v_1 \right), \left( s, v_{2j} \right), \left( v_1, t \right), \left( v_{2j}, t \right) \rbrace$ as the set of the outermost edges of $G^B_j$. For $e \in E_j$, we assign (state-independent) edge costs $c^{j}_{e} (x) = a^j_{e} \, x + b^j_{e}$ as follows:
	\begin{align*}
		a^j_{e} &= \begin{cases}
			2 \cdot 10^{j-1} & \text{if}\ e \in E^1_j \setminus \lbrace (v_j, v_{j+1}) \rbrace, \\ 
			0 & \text{else}
		\end{cases}   \\[4mm] b^j_{e} &= \begin{cases}
			10^{j-1-i} & \text{if}\ e =(v_i, v_{2j-i}) \ \text{or}\ e = (v_{i+1}, v_{2j+1-i}), \ \text{for} \ i=0,...,j-1, \\
			0 & \text{else.}
		\end{cases}
	\end{align*}
    In comparison to \cite{KlimmW22}, we scaled all $a_e^j$ by a factor $2 \cdot 10^{j-1}$. By doing so, demands $\tilde{d}_1,...,\tilde{d}_{2^{j+1}} \in [0,1]$ are now sufficient to generate the $2^{j+1}$ different support sets that emerge in their construction. 

    Note that $a^j_e = 0$ for some cost functions, which we technically do not allow, as we assume all $c_e(x)$ to be strictly increasing. As argued in \cite{KlimmW22}, the construction works with slopes $a^j_e = \epsilon$ added to those cost functions, for sufficiently small $\epsilon \ll 1$. We omit these $\epsilon$-slopes for the benefit of presentation. 
    
    Let $\Theta = \lbrace \theta_1, \theta_2 \rbrace$. We set $c^{j,\theta}_e(x) = c^j_e(x)$ for all $e \in E_j$ and $\theta \in \Theta$, i.e., all cost functions have the same offset in both states. We complete our construction by adding a direct path $p = (s,t)$ to $G^B_j$ for which $c^{j,\theta_1}_p(x) = 0$ and $c^{j,\theta_2}_p(x) = 3 \cdot 10^{j-1}$. Thus, for a given distribution $\mu_\alpha := (1-\alpha,\alpha)^\top$, $\alpha \in [0,1]$, the expected cost of edge $p$ is $c^j_{p}(x \mid \mu_\alpha) = 3 \cdot 10^{j-1} \alpha$. Hence, whenever edge $p$ is active, any other path (through $G^B_j$) in the Wardrop flow cannot induce lower cost. Figure~\ref{fig:mixed_braess_j=3} shows a sketch of the resulting graph $G^{B+p}_j$ for $j=3$. 
    We define $\rho(\mu_\alpha) \in [0,1]$ as the fraction of infinitesimally small agents that routes via a path through $G^B_j$ given $\mu_\alpha$. 
        
    We show that $\rho(\mu_\alpha)$ is strictly increasing and continuous in $\mu_\alpha$. We follow \cite[Claim 4]{KlimmW22} and assert the following properties for our construction:
    \begin{itemize}
        \item[(i)] If $\alpha = 0$, then ${c}^j_{p}(x \mid \mu_\alpha) = 0$ and $\rho(\mu_\alpha) = 0$.
        \item[(ii)] If $\alpha \geq \frac{2}{3}$, then $\rho(\mu_\alpha) = 1$, $x_e(\mu_\alpha) = 0$ for all $e \not \in E^O_j$, and $\pi_t(\mu_\alpha) - \pi_s(\mu_\alpha) = 2 \cdot 10^{j-1}$.
    \end{itemize}
    Statement (i) is evident. For (ii), we obtain $x_e(\mu_\alpha) = 0$ for all $e \not \in E^O_j \cup \lbrace p \rbrace$, and $\pi_t(\mu_\alpha) - \pi_s(\mu_\alpha) = 2 \cdot 10^{j-1}$ from the same arguments as made in \cite{KlimmW22}. Furthermore, $\pi_t(\mu_\alpha) - \pi_s(\mu_\alpha) = 2 \cdot 10^{j-1} \leq  3 \cdot 10^{j-1} \alpha$ for $\alpha \geq \frac{2}{3}$. Hence, for $\mu_\alpha$ with $\alpha \geq \frac{2}{3}$, the edges in $E^O_j$  are used exclusively. 
     
    By Corollary~\ref{cor:monotone}, the per-unit cost $\pi_t(d)$ of the WE as function of the demand $d \in [0,1]$ is strictly increasing for $G^B_j$. In addition, as shown by \citet[Proposition 3.1]{CominettiDS19}, $\pi_t(d)$ is continuous. Since  $c^j_p (x \mid \mu_\alpha) = 3 \cdot 10^{j-1} \alpha$ is continuous and strictly increasing in $\alpha$ as well, this implies the existence of a bijective mapping between all $\alpha \in [0,\frac{2}{3}]$ and all demands $\rho(\mu_\alpha) \in [0,1]$ for $G^{B+p}_j$. In particular, it follows that there exist distinct parameter values $\alpha_1,...,\alpha_{2^{j+1}} \in [0,\frac{2}{3}]$ that induce demand rates $\rho(\mu_{\alpha_1}),..., \rho(\mu_{\alpha_{2^{j+1}}})$, which again give rise to supports $A^1,...,A^{2^{j+1}}$. 
    Thus, there are $2^{j+1}$ different supports.
\end{proof}

\section{Illustration for \Cref{{sec:compStudy}}}
\label{app:illu_sf}
  
  \input{SF.tex}

\end{document}